\documentclass[journal]{IEEEtran}
\usepackage{amsmath,amssymb}
\usepackage{graphicx,graphics,color,psfrag}
\usepackage{diagbox}
\usepackage{cite,balance}
\usepackage{makecell}
\usepackage{caption}
\usepackage{framed}
\allowdisplaybreaks
\usepackage{algorithm}
\usepackage{algorithmic}
\usepackage{accents}
\usepackage{amsthm}
\usepackage{bm}
\usepackage{url}
\usepackage{bbding}
\usepackage[english]{babel}
\usepackage{multirow}
\usepackage{enumerate}
\usepackage{cases}
\usepackage{stfloats}
\usepackage{dsfont}
\usepackage{color,soul}
\usepackage{amsfonts}
\usepackage{cite,graphicx,amsmath,amssymb}
\usepackage{fancyhdr}
\usepackage{hhline}
\usepackage{graphicx,graphics}
\usepackage{array,color}
\usepackage{mathtools}
\usepackage{amsmath}
\usepackage[T1]{fontenc}
\usepackage{float}  
\usepackage{stfloats}
\usepackage{siunitx}
\usepackage{subfloat}
\usepackage[caption=false,font=footnotesize,labelfont=rm,textfont=rm]{subfig}

\def\BibTeX{{\rm B\kern-.05em{\sc i\kern-.025em b}\kern-.08em
		T\kern-.1667em\lower.7ex\hbox{E}\kern-.125emX}}
\usepackage{balance}

\newtheorem{lemma}{\emph{\underline{Lemma}}}

\newtheorem{remark}{\bf \emph{\underline{Remark}}}

\def\rank{{\operatorname{rank}}}

\def\({\left(}
\def\){\right)}

\def\b0{{\mathbf{0}}}

\newcommand{\diag}{\mathrm{diag}}

\begin{document}
\captionsetup[figure]{name={Fig.}}
\title{Mixed Near-field and Far-field Localization in Extremely Large-scale MIMO Systems}
\author{Cong Zhou, Changsheng You,~\IEEEmembership{Member,~IEEE}, Chao Zhou, Hongqiang Cheng, Shuo Shi,~\IEEEmembership{Member,~IEEE}
	\thanks{Cong Zhou is with the Department of Electronic and Electrical Engineering, Southern University of Science and Technology, Shenzhen 518055, China and also with the School of Electronics and Information Engineering, Harbin Institute of Technology, Harbin, 150001, China (e-mail:  zhoucong@stu.hit.edu.cn).}
	\thanks{Changsheng You, Chao Zhou and Hongqiang Cheng are with the Department of Electronic and Electrical Engineering, Southern University of Science and Technology, Shenzhen 518055, China (e-mails: youcs@sustech.edu.cn, zhouchao2024@mail.sustech.edu.cn, chenghq2023@mail.sustech.edu.cn).}
	\thanks{Shuo Shi is with the School of Electronics and Information Engineering, Harbin Institute of Technology, Harbin, 150001, China (e-mail: crcss@hit.edu.cn).}
	\thanks{\emph{Corresponding author: Changsheng You.}}
}
\maketitle
\begin{abstract}
    In this paper, we study efficient \emph{mixed near-field and far-field} target localization methods in extremely large-scale multiple-input multiple-output (XL-MIMO) systems.
    Compared with existing works, we address two new challenges in target localization of MIMO communication systems via using decoupled subspace methods, arising from the \emph{half-wavelength antenna spacing} constraint and \emph{hybrid uniform planar array} (UPA) architectures.
    To this end, we propose a new three-step mixed-field localization method. 
    First, we reconstruct the equivalent signals received at UPA antennas by judiciously designing analog combining matrices over time with minimum recovery errors.
    Second, based on recovered signals, we extend the modified multiple signal classification (MUSIC) algorithm to the UPA architectures by constructing a new covariance matrix of a virtual sparse UPA (S-UPA) to decouple the 2D angles and range estimation.
    Due to the structure of the S-UPA, there exist \emph{ambiguous angles} when estimating true angles of targets.
    In the third step, we design an effective classification method to distinguish mixed-field targets, determine true angles of all targets, as well as estimate the ranges of near-field targets.
    In particular, angular ambiguity is resolved by showing an important fact that the three types of estimated angles (i.e., far-field, near-field, and ambiguous angles) exhibit significantly different patterns in the range-domain MUSIC spectrum. 
    Furthermore, to characterize the estimation error lower-bound, we obtain a matrix closed-form Cram\'er-Rao bounds for mixed-field target localization.
    Finally, numerical results demonstrate the effectiveness of our proposed mixed-field localization method, which improves target-classification accuracy and achieves a lower root mean square error than various benchmark schemes.
\end{abstract}
  
\section{Introduction}
Sensing functionality at the base station (BS) has been listed by International Telecommunication Union Radio communication Sector (ITU-R) as one of the six major usage scenarios in future sixth-generation (6G) wireless networks, to accommodate a wide range of emerging applications such as smart home, low-altitude economy (LAE), autonomous driving~\cite{zhang2023dual, you2024next}.
Among various design paradigms, communication-centric designs has received upsurging research interests in both academia and industry, which effectively leverage existing wireless infrastructures (e.g., BSs, wireless backhaul) for enabling ubiquitous and high-resolution sensing. Following this paradigm, we study in this paper how to leverage \emph{extremely large-scale multiple-input multiple-output (MIMO)} (XL-MIMO) systems to achieve \emph{mixed near-field and far-field} target localization \cite{Lu_NFtutorial}.
In particular, we propose efficient methods to address a set of localization challenges by using decoupled subspace methods encountered in XL-MIMO communication systems, which arise from the half-wavelength antenna spacing constraint and the hybrid uniform planar array (UPA) architecture.

\vspace{-5pt}
\subsection{Related Works}
Accurate and wide-range localization lays important foundations for various tasks in cellar networks such as beam management, user tracking, etc.
To achieve this, one promising approach is by leveraging XL-MIMO systems to expand the sensing coverage and achieve high sensing resolution-and-accuracy in both near-field and far-field.
Compared with conventional massive MIMO systems, XL-MIMO increases the number of antennas by another order-of-magnitude (e.g., from 256 to 1024), thus substantially enhancing system spectrum efficiency and spatial resolution. Moreover, the \emph{Rayleigh distance}, which specifies the boundary between radiative near-field and far-field, is greatly increased. For example, when the system operates at $28$ GHz frequency bands, the Rayleigh distance of an 1-meter (m) array is up to 187 m. As such, in future XL-MIMO systems, there may appear the case where some communication users/sensing targets are located in the near-field of the BS, while the others are in the far-field. This thus necessaries efficient designs for mixed-field ISAC in XL-MIMO systems. 

For the communication perspective, most of recent research on XL-MIMO focused on near-field only systems, including e.g., spherical and non-stationary channel modelling \cite{chen2023non}, power scaling laws \cite{LuCommunicating2022}, beamfocusing/beamforming design, channel estimation and beam training~\cite{wu2024near}; while some initial efforts have been devoted to studying mixed-field communication designs such as channel estimation \cite{DaiHF_CE}, interference management \cite{YouHF_IM}, beam scheduling and power control \cite{YouHF_SWIPT}.
On the other hand, for XL-MIMO target sensing/localization, the existing research can be generally classified into three categories, namely, far-field, near-field, and mixed-field localization.

Specifically, for \emph{far-field localization}, various algorithms have been proposed to estimate target angles using typical communication arrays.
Taking digital arrays of half-wavelength inter-antenna spacing as an example, efficient subspace methods such as multiple signal classification (MUSIC) \cite{schmidt1986multiple} and estimating signal parameters via rotational invariance techniques (ESPRIT) \cite{roy1989esprit} have been developed to estimate target angle, which essentially exploit orthogonality between signal and noise subspaces to distinguish impinging sources. Further, for UPA systems,  the one-dimensional (1D) MUSIC and ESPRIT algorithms were extended in \cite{chen1997spatial} and \cite{wang2005comments} to joint estimate two-dimensional (2D) angular parameters (i.e., elevation and azimuth angles). To reduce the computational complexity arising from 2D exhaustive search in 2D MUSIC algorithms, the authors in \cite{fan2018angle} proposed to first exploit a 2D discrete Fourier transform (DFT) algorithm to estimate coarse 2D angles and then refine angle estimation resolution by using angle rotation techniques \cite{fan2018angle}.

Next, for \emph{near-field localization},  spherical wavefronts provide a new degree-of-freedom (DoF) in range estimation, hence allowing for joint estimation of both the target angle and range even in the narrow-band, yet without deploying multiple anchors with strict network synchronization \cite{11146870}. This can be achieved by designing 2D (angle-and-range domain) MUSIC algorithm for uniform linear array (ULA) systems and 3D (elevation/azimuth angles and range domain) for UPA systems.
As 2D/3D exhaustive grid-search needs to be performed for target location, these algorithms  incur prohibitively high computational complexity in practice.
To address this issue, the authors in \cite{zhang2018localization} assumed that the inter-antenna spacing is no larger than \emph{one-quarter wavelength}, hence avoiding intricate \emph{angular ambiguity}.
Based on this, they proposed a reduced-dimensional MUSIC (RD-MUSIC) algorithm, which only requires 1D MUSIC spectrum search in the angular domain via splitting  array steering matrix in terms of angle and range. Besides, by extracting anti-diagonal elements of the original covariance matrix, an effective modified MUSIC algorithm was proposed in \cite{he2011efficient} (yet still based on \emph{one-quarter} wavelength antenna spacing), which decomposed 2D (angle-and-range domain) peak search into two 1D grid-search in the angular and range domains, respectively.
Moreover, for UPA systems, the authors in \cite{huang2023low} leveraged the anti-diagonal elements of signal covariance matrix to decouple the 3D MUSIC algorithm into 2D angle and 1D range estimation problems, whose complexity was further reduced by exploiting  2D DFT algorithm and angle rotation technique in \cite{fan2018angle}. 

{Last, for \emph{mixed-field} localization, due to the coupling between received spherical and planar wavefronts, it is necessary to differentiate far-field and near-field targets. To achieve this, for ULA systems, the authors in \cite{liu2013efficient, he2011efficient} proposed to firstly identify the angle of far-field sources by finding the peak of the 2D (angle-and-range) MUSIC spectrum in the angular domain via setting the range parameter $r\to\infty$ \cite{liu2013efficient, he2011efficient}. Then, far-field signals were removed and the locations of near-field users were estimated from residual signals by using typical methods as for near-field localization.
The authors in~\cite{TS_MUSIC} addressed the mixed near-field and far-field target localization problem by a two-stage MUSIC algorithm, which, however, incurs high computational complexity. 
Furthermore, an enhanced approach was proposed in~\cite{TS_MD} by applying a matrix auto-differencing operation to the fourth-order cumulant (FOC) matrix for achieving superior classification and localization performance in such mixed scenarios.
In view of the above works on mixed-field localization, there still exist several  limitations as listed below.
\begin{itemize}
	\item {\bf{(Half-wavelength inter-antenna spacing)}} For low-complexity mixed-localization design using subspace methods, the existing works rely on a strong assumption that the inter-antenna spacing is no larger than one-quarter wavelength to avoid angular ambiguity, which, however, is incompatible to legacy wireless systems with typically half-wavelength spacing. 
	Moreover, larger number of ambiguous angles arising from \emph{half-wavelength inter} constraints also further complicate the differentiation between far-field and near-field targets.
	\item {\bf{(Hybrid UPA architecture)}} Existing angle/location estimation methods mostly considered fully digital architectures. Nevertheless, these methods may not be applicable to hybrid UPA systems (widely deployed in high-frequency wireless communication and/or ISAC systems), due to the reduced sensing DoFs after analog combining as well as more complex signal processing techniques for UPAs. 
\end{itemize} 
\vspace{-10pt}
\subsection{Contribution and Organization}
To address the above issues, in this paper, we consider an XL-MIMO target localization system for mixed-filed targets as shown in Fig.~\ref{fig:system model}, where the BS is equipped with a hybrid UPA array with \emph{half-wavelength} inter-antenna spacing. Our target is to employ the XL-MIMO BS to localize both near-field and far-field targets, where the targets actively send probing signals to facilitate localization at the BS. The main contributions of this paper are summarized as follows.

First, we propose a new and efficient framework to localize both near-field and far-field targets, which comprises three main steps.
In Step 1, we judiciously design the analog combining matrices at the BS over time to recover the equivalent signals (like a fully-digital architecture) received at the UPA with minimum recovery errors.
Then, in Step 2, we extend the reduced-dimensional MUSIC algorithm to the UPA scenario by constructing an equivalent far-field covariance matrix of a \emph{virtual sparse UPA} (S-UPA) from the original covariance matrix, which decouples the joint 3D angle-and-range estimation into sequential angle and range estimations.
Subsequently, a decoupled 2D angular MUSIC algorithm is developed to further separate the estimation of elevation and azimuth angles, obtaining candidate 2D angles for all targets, including both true far-field and near-field angles as well as the corresponding ambiguous angles induced by the virtual S-ULA.
In Step 3, we design an effective classification method to distinguish mixed-field targets, determine true angles of all targets, as well as estimate the ranges of near-field targets.
In particular, we resolve angular ambiguity by showing an important fact that the three types of estimated angles have significantly different patterns in the range-domain MUSIC spectrum under the half-wavelength inter-antenna spacing constraint.
In addition, to evaluate the lower-bound of the proposed algorithm, we characterize the Cram\'er-Rao Bounds (CRB) for mixed-field target localization.
Moreover, extensive numerical results are presented to demonstrate the effectiveness of our proposed mixed-field target localization method.
It is shown that the proposed algorithm can achieve nearly the same performance with the 3D MUSIC algorithm, while significantly reducing its computational complexity.
In addition, due to larger array aperture, the proposed mixed-field localization method achieves a lower root mean square error (RMSE) than UPAs with one-quarter wavelength inter-antenna spacing.
\textit{Notations}: Lowercase and uppercase boldface letters are used to represent vectors and matrices, such as $\mathbf{a}$ and $\mathbf{A}$.
For vectors and matrices, the symbol $ (\cdot)^{H} $  indicates the conjugate transpose operation, while the symbol $[\mathbf{A}]_{i,j}$ and $[\mathbf{a}]_{n}$ denote the $(i,j)$-th entry of the matrix $\mathbf{A}$ and the $n$-th entry of the vector $\mathbf{a}$, respectively.
Additionally, calligraphic letters are employed to denote sets, while the symbol $\mathbb{Z}$ represents the integer set.
The notations $ \left|\cdot\right| $ and $ \left\lVert \cdot\right\lVert $ refer to the absolute value of a scalar and the $ \ell_{2} $ norm, respectively.
The symbol $\mathbf{I}_K $ represents a $ K $-dimensional identity matrix, while the symbol $ \text{angle}(x) $ denotes the phase of the complex value $ x $.
Finally, the notation $\mathcal{CN}(\mu,\sigma^2)$ represents the complex Gaussian distribution with mean $\mu$ and variance $\sigma^2$
\begin{figure*}[t]
	\includegraphics[width=0.75\textwidth]{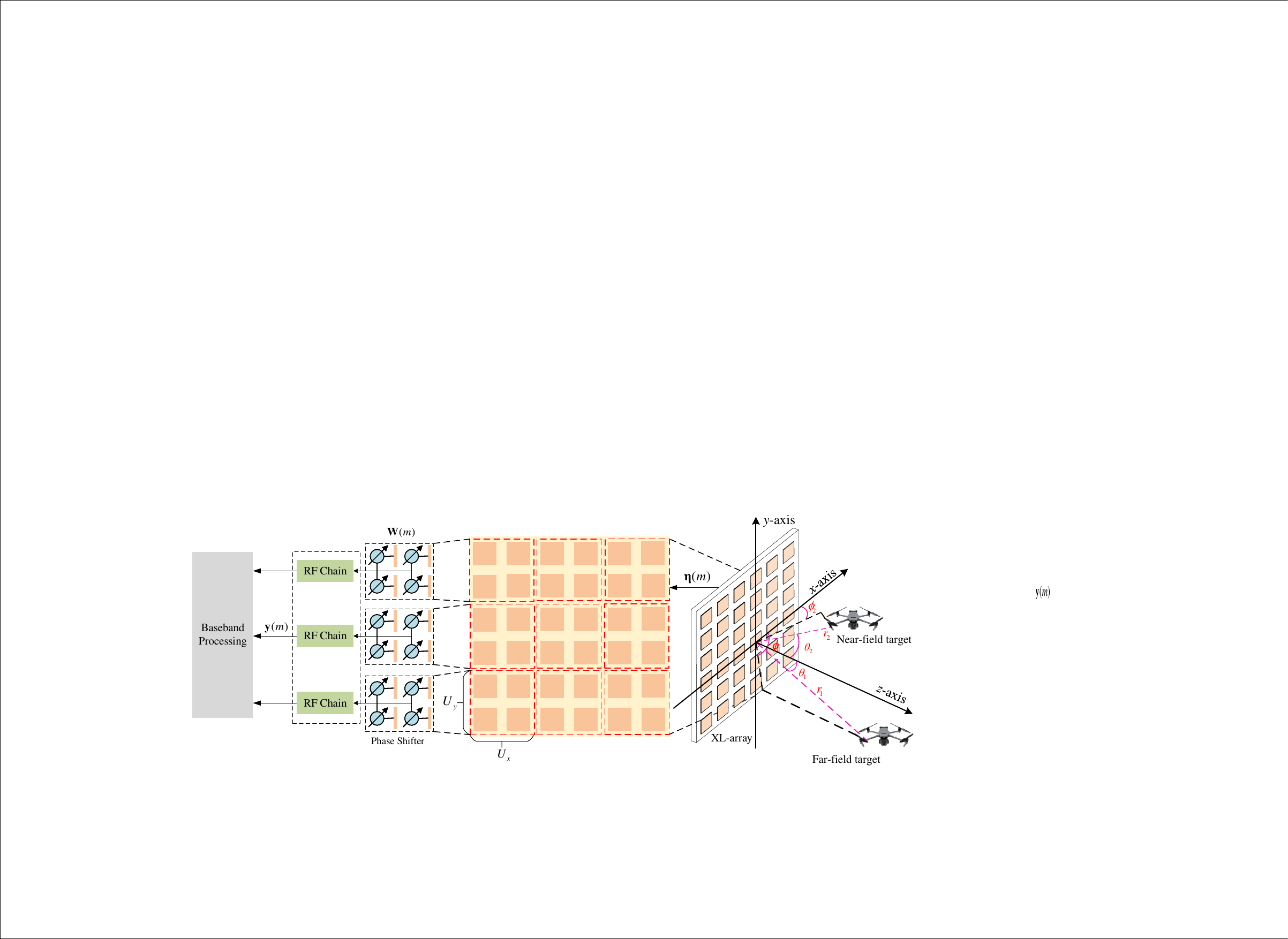}
	\centering
	\caption{Considered mixed-field target localization system, where the BS is equipped with a (sub-connected) hybrid UPA.}
	\label{fig:system model}
	\vspace{-12pt}
\end{figure*}
\vspace{-5pt}
\section{System Model}
We consider an XL-MIMO wireless sensing system shown in Fig.~\ref{fig:system model}, where a BS equipped with a UPA is employed for localizing both near-field and far-field targets. The BS is placed in the $ xy $-plane and centered at the origin, which consists of $ N = N_xN_y $ antennas with $ N_x $ and $ N_y $ denoting the number of antennas along the $ x $-axis and $ y $-axis, respectively. Different from existing literature for mixed-field target localization assuming fully-digital architectures, we consider typical wireless communication infrastructures in this paper. First, the inter-antenna spacing of UPA along the $ x $-axis  and  $ y $-axis is \emph{half-wavelength}, which are respectively denoted by  $ d_x=d_y=d_0 = \frac{\lambda}{2} $, with $ \lambda $ denoting the carrier wavelength. Second, the BS  employs a \emph{hybrid array architecture}, for which $ N $ antennas are connected to $ N_{\rm RF} \ll N $ radio frequency (RF) chains with each one connecting $ U \triangleq  {N}/{N_{\rm RF}} = M_xM_y $ phase shifters as illustrated in Fig. \ref{fig:system model} (referred to as sub-connected architecture). Here, $ U_x $ and $ U_y $ are the number of PSs (assumed to be integers) along the $ x $-axis and $ y $-axis, respectively, and 
we assume $U = {N}/{N_{\rm RF}} $ being an integer without loss of generality. Moreover, we consider the passive sensing scenario (i.e., the BS does not actively transmit probing signals) where both near-field and far-field targets transmit uncorrelated signals to facilitate precise localization at the BS, while the extension to BS active sensing case will be discussed in Section \ref{sec:Extensions and Discussions}. 
\vspace{-10pt}
\subsection{Channel Models}
For the BS, the $ n $-th antenna is located at $ (n_xd_0, n_yd_0, 0) $, where $ n_x \in \mathcal{N}_x \triangleq \{-\frac{N_x-1}{2},\cdots,\frac{N_x-1}{2}\} $ and $ n_y \in \mathcal{N}_y \triangleq \{-\frac{N_y-1}{2},\cdots,\frac{N_y-1}{2}\} $ denote the indices of antennas along the $ x $-axis and $ y $-axis, respectively. As such, the antenna index $n$ can be represented as $ n = n_x+n_yN_x+\frac{N+1}{2} \in \mathcal{N} \triangleq \{1,2,\cdots,N\}$, and the Rayleigh distance of the UPA is given by $ Z_{\rm R} =2(N_x^2d_0^2+N_y^2d_0^2)/{\lambda} $.
Let $K$ denote the number of targets, consisting of $ K_1 $ targets in the far-field and $ K_2 $ targets in the near-field of the UPA, whose channels are modeled as follows.

\underline{\bf{Far-field channel model:}}
For far-field target $k \in \mathcal{K}_1 \triangleq \{1,2,\cdots, K_1\} $, its range with the BS  is  larger than the Rayleigh distance $ Z_{\rm R}$.
For target localization in high-frequency bands, we mainly focus on the line-of-sight (LoS) channel path as the NLoS channel paths exhibit negligible power due to severe path-loss and shadowing \cite{wu2024near}.
Hence, the far-field channel model from target $k$ to the BS, denoted by $\mathbf{h}_{k}\in\mathbb{C}^{N\times 1}$, is given by
\begin{equation}\label{Eq:far-model}
	\mathbf{h}_{k}= \sqrt{N}h_{k} \mathbf{a}(\theta_{{k}},\phi_{{k}}), \forall k \in \mathcal{K}_1,
\end{equation}
where $ h_{k} = \frac{\sqrt{\beta_0}}{r_{k}} $ denotes the complex channel gain with $ \beta_0 $ representing the reference channel power gain at a range of 1 m. 
Moreover, $\mathbf{a}(\theta_{{k}},\phi_{{k}}) $ denotes the far-field steering vector with $ \theta_{{k}} \in [0,\frac{\pi}{2}] $ and $ \phi_{{k}}\in [-\pi,\pi] $ denoting the elevation and azimuth angles of the $ k $-th target ($k\in\mathcal{K}_1$), respectively.
In particular, $ \mathbf{a}(\theta_{{k}},\phi_{{k}}) $ is given by
\begin{equation}
	\mathbf{a}(\theta_{{k}},\phi_{{k}}) = \mathbf{a}_y(\alpha_{{k}})\otimes \mathbf{a}_x(\beta_{{k}}), \forall k \in \mathcal{K}_1,
\end{equation}
where $ \mathbf{a}_x(\alpha_{{k}}) \in \mathbb{C}^{N_x\times1}$ and  $ \mathbf{a}_y(\beta_{{k}}) \in \mathbb{C}^{N_y\times1}$ represent the channel response vectors for the antennas in the $ x $-axis and $ y $-axis, respectively.
Mathematically, we have $ [\mathbf{a}_x(\alpha_{{k}})]_{n_x+\frac{N_x+1}{2}} = e^{\jmath\frac{2\pi}{\lambda}d_0 n_x \alpha_{k} }$ and $ [\mathbf{a}_y(\beta_{{k}})]_{n_y+\frac{N_y+1}{2}} = e^{\jmath\frac{2\pi}{\lambda}d_0 n_y \beta_{k} }$, where parameters $ \alpha_{k} $ and $ \beta_{k} $ are defined as $$ \alpha_{k} \triangleq \sin\theta_{{k}}\cos\phi_{{k}}, ~~\beta_{k} \triangleq \sin\theta_{{k}}\sin\phi_{{k}}. $$ 

\begin{figure*}[t]
	\includegraphics[width=0.9\textwidth]{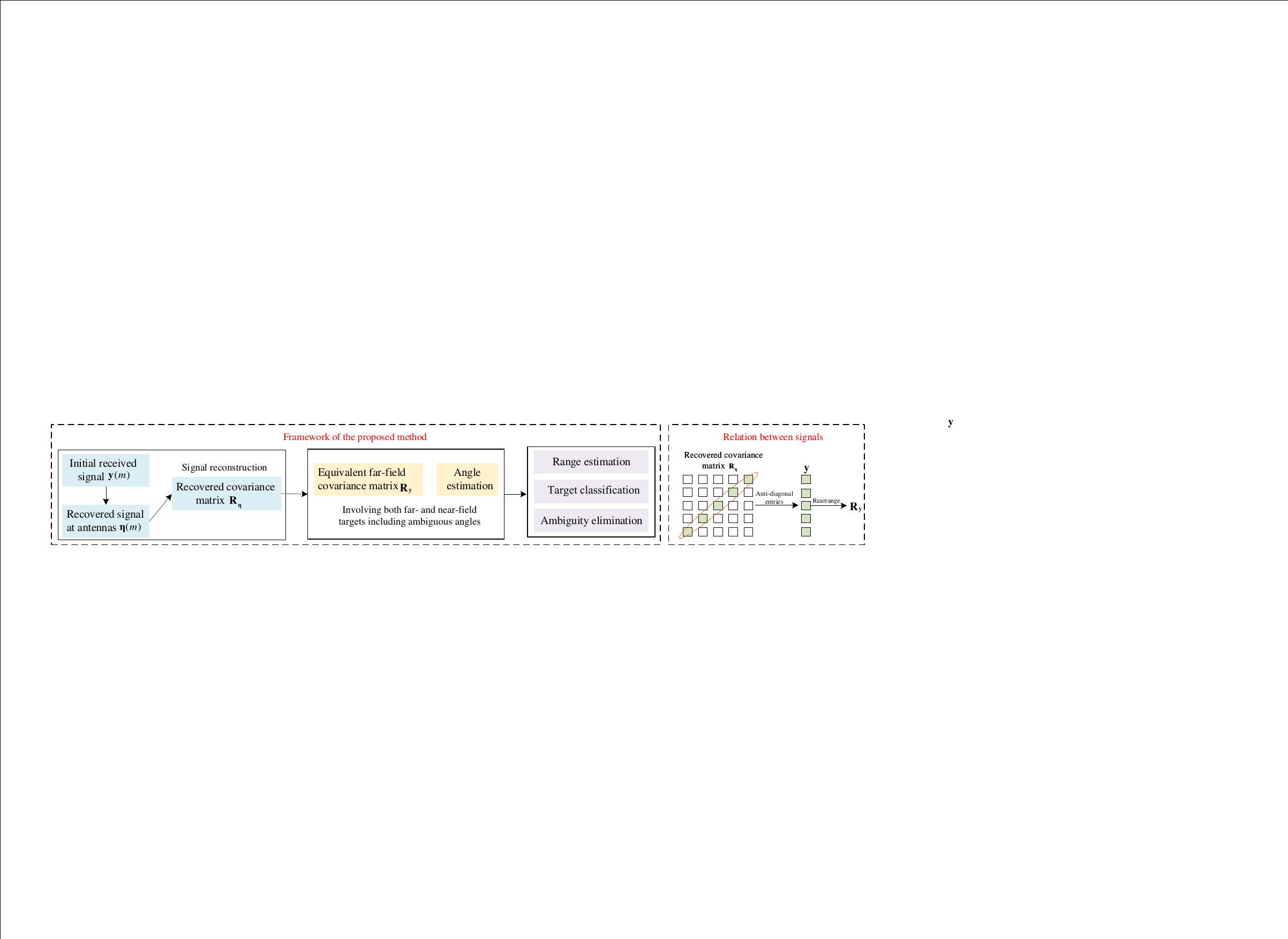}
	\centering
	\caption{The framework of proposed mixed-field localization algorithm and relation between signals.}
	\label{fig:Algorithm}
	\vspace{-12pt}
\end{figure*}
\underline{\bf{Near-field channel model:}}
For near-field target $ k \in \mathcal{K}_2 \triangleq \{K_1+1,K_1+2,\cdots, K\} $, its range $ r_{k} $ is smaller than the Rayleigh distance $Z_{\rm R}$.
As such, the near-field LoS channel from target $ k \in \mathcal{K}_2$ to the BS can be modeled as  \cite{Zhang2022fast}
\begin{equation}\label{Eq:nf-model}
	\mathbf{h}_{k}= \sqrt{N}h_{k} \mathbf{b}(r_{{k}}, \theta_{{k}}, \phi_{{k}}), \forall k \in \mathcal{K}_2,
\end{equation}
where $ \mathbf{b}(r_{{k}}, \theta_{{k}}, \phi_{{k}}) $ denotes its near-field channel steering vector.
Mathematically, $ \mathbf{b}(r_{{k}}, \theta_{{k}}, \phi_{{k}}) $ can be modeled as
\begin{equation}
	\label{near_steering}
	\left[\mathbf{b}\left(r_{{k}},\theta_{{k}}, \phi_{{k}}\right) \right]_{n} = \frac{1}{\sqrt{N}}e^{-\frac{\jmath 2 \pi (r_{k,n}-r_{k})}{\lambda}}, \forall n\in \mathcal{N}, 
\end{equation}
where $ r_{k,n} $ denotes the range from the near-field target $ k\in \mathcal{K}_2 $ to the $ n $-th antenna.
Mathematically, $ r_{k,n} $ can be expressed as
\begin{equation}\label{Eq:r_kn}
	r_{k,n} = r_{k} \Big({1\!+\!\frac{\left(n_x^2+n_y^2\right) d_0^2}{r_{k}^2}\!-\!\frac{2 d_0}{r_{k}}\left(n_x\alpha_{k}+n_y \beta_{k}\right)}\Big)^{\frac{1}{2}}\!, 
\end{equation}
for $k\in \mathcal{K}_2$, where $ r_{k} $ represents the range between the target and the center of the BS.  
Further, by means of two-dimensional Fresnel approximation, $ r_{k,n} $ in \eqref{Eq:r_kn} can be approximated as \cite{pan2022localization}
\begin{equation}
	\begin{aligned}
		&r_{k,n} \approx r_{k}-n_xd_0\alpha_{k}-n_yd_0\beta_{k} \\
		&~~~~\!+\!\frac{n_x^2 d_0^2 (1- \alpha_{k}^2)\!+\! n_y^2 d_0^2 (1- \beta_{k}^2)}{2 r_{k}} \!-\! \frac{n_xn_yd_0^2\alpha_{k}\beta_{k}}{ r_{k}}.
	\end{aligned}
\end{equation} 
Then, the near-field channel steering vector for target $ k\in \mathcal{K}_2 $ can be rewritten as
\begin{equation}
	\label{eq:new near_steering}
	\left[\mathbf{b}\!\left(r_{k},\alpha_{k},\beta_{k}\right) \right]_{{n}} \!=\! \frac{1}{\sqrt{N}}e^{\jmath\frac{2\pi}{\lambda} \big(u_n(\alpha_{k},\beta_{k}) + v_n(r_{k},\alpha_{k},\beta_{k})\big)}, 
\end{equation}
where $ u_n(\alpha_{k},\beta_{k}) \triangleq n_xd_0\alpha_{k}+n_yd_0\beta_{k} $ contains the angle information only, and $ v_n(r_{k},\alpha_{k},\beta_{k}) \triangleq \frac{n_xn_yd_0^2\alpha_{k}\beta_{k}}{ r_{k}} - \frac{n_x^2 d_0^2 (1- \alpha_{k}^2)+n_y^2 d_0^2 (1- \beta_{k}^2)}{2 r_{k}} $ is jointly determined by both the target angles and range.
\vspace{-0.2cm}

\subsection{Signal Model}
The $ K $ targets are assumed to simultaneously transmit pilots in $ M $ time slots which are in  narrow-band and uncorrelated.
As such, the received signal in the UPA antennas at time $m$, denoted by $\boldsymbol{\eta}(m) \in \mathbb{C}^{N\times 1}$, is given by
\begin{equation}
	\begin{aligned}
	\boldsymbol{\eta}(m) = &\mathbf{G}\mathbf{s}(m) +  \mathbf{n}(m) \\
	=&\mathbf{A} \mathbf{s}_{{\rm F}}(m) +\mathbf{B} \mathbf{s}_{{\rm N}}(m)+\mathbf{n}(m),
	\end{aligned}
	\label{eq:received signal at antennas}
\end{equation}
where $ \mathbf{G} \triangleq [\mathbf{A}, \mathbf{B}] \in \mathrm{C}^{N \times K} $ denotes the array steering matrix (ASM) with $ \mathbf{A} \triangleq [\mathbf{a}\left(\theta_{1}, \phi_{1}\right),\mathbf{a}\left(\theta_{2}, \phi_{2} \right),\cdots,\mathbf{a}\left(\theta_{K_1}, \phi_{K_1}\right)]  $ and $ \mathbf{B} \triangleq [\mathbf{b}\left(r_{K_1+1},\theta_{K_1+1}, \phi_{K_1+1}\right),\cdots,\mathbf{b}\left(r_{K},\theta_{K}, \phi_{K}\right)] \in \mathbb{C}^{N \times K_2} $ denoting the ASM for the far- and near-field targets, respectively.
The noise vector is distributed as $ \mathbf{n}(m) \sim \mathcal{C \mathcal { N }}\left(0, \sigma^2 \mathbf{I}_M\right) $ with zero mean and variance $ \sigma^2 $. Moreover, $ \mathbf{s}(m) = [\mathbf{s}^T_{\rm F}(m), \mathbf{s}^T_{\rm N}(m)]^T \in \mathbb{C}^{K \times 1} $ denotes the source vector (i.e., transmitted signals by targets) with $ \mathbf{s}_{\rm F}(m) = [s_1(m),s_2(m),\cdots,s_{K_1}(m)]^T \in \mathbb{C}^{K_1 \times 1}$ and $ \mathbf{s}_{\rm N}(m) = [s_{K_1+1}(m),s_{K_1+2}(m),\cdots,s_{K}(m)]^T \in \mathbb{C}^{K_2 \times 1}$ representing the source sub-vectors for far- and near-filed targets, satisfying
\begin{equation}
	\mathbb{E}\{s_j(m)s_k(m)\}=\left\{\begin{matrix}
	g_k,~~j=k,\\ 
	0,~~j\neq k,
	\end{matrix}\right.
\end{equation}
where $ p_k $ and $  g_k = N|\beta_k|^2p_k $ denote the transmit power and equivalent channel power gain for the $ k $-th target, respectively.

After analog combining, the received signal at discrete time $ m $, denoted by $ \mathbf{y}(m)\in\mathbb{C}^{N_{\rm RF}\times 1}$,
is given by
\begin{equation}
	\begin{aligned}
		&\mathbf{y}(m) \!\!=\!\! \mathbf{W}(m)\boldsymbol{\eta}(m)\\
        &~~~~=\mathbf{W}(m)\mathbf{A} \mathbf{s}_{{\rm F}}(m) \!\!+\!\! \mathbf{W}(m)\mathbf{B} \mathbf{s}_{{\rm N}}(m) \!\!+\!\! \mathbf{W}(m)\mathbf{n}(m),
	\end{aligned}
	\label{eq:received signal}
\end{equation}
where $ \mathbf{W}(m) \in \mathbb{C}^{N_{\rm RF}\times N} $ denotes the analog combining matrix. For the considered sub-connected hybrid beamforming architecture in Fig.~\ref{fig:system model}, $\mathbf{W}(m) $ can be expressed as
\begin{equation}\label{Eq:W}
	\mathbf{W}^T(m) \!=\! \boldsymbol{\Pi} \mathbf{V}(m)  
	\!=\! \boldsymbol{\Pi} \left[\begin{array}{cccc}
	\!\!\!\mathbf{v}_{1}(m) & \mathbf{0} & \ldots & \mathbf{0}\! \\
	\mathbf{0} & \!\mathbf{v}_{2}(m)\! & \ldots & \mathbf{0}\! \\
	\vdots & \vdots & \ddots & \vdots \\
	\mathbf{0} & \mathbf{0} & \ldots & \mathbf{v}_{\rm RF}(m)\!\!\!\!
	\end{array}\right]\!\!,
\end{equation}
where $ \boldsymbol{\Pi} \in \mathbb{C}^{N \times N} $ and $ \mathbf{v}_{f}(m) \in  \mathbb{C}^{U \times 1}, f \in \mathcal{F} \triangleq \{1,2,\cdots, N_{\rm RF}\} $ denote the permutation matrix determined by the sub-connected configuration and the combining vector of the $f$-th RF chain, respectively.
For example,  consider a UPA with $ N_x = 2 $ and $ N_y = 4 $ with $ N_{\rm RF} = 2 $ RF chains. Following the antennas' indices, the combining vectors of the first and second RF chains are given by $ [\mathbf{W}^T(m)]_{:,1} = [v_{1,1}(m),v_{1,2}(m),0,0,v_{1,3}(m),v_{1,4}(m),0,0]^T$, and $ [\mathbf{W}^T(m)]_{:,2} = [0,0,v_{2,1}(m),v_{2,2}(m),0,0,v_{2,3}(m),v_{2,4}(m)]^T $, respectively.
In addition, the permutation matrix $ \boldsymbol{\Pi} $ can be constructed as $ \boldsymbol{\Pi} = [\mathbf{e}_1,\mathbf{e}_2,\mathbf{e}_5,\mathbf{e}_6,\mathbf{e}_3,\mathbf{e}_4,\mathbf{e}_7,\mathbf{e}_8]^T $, where $ \mathbf{e}_n $ denotes the $ n $-th column of the identity matrix $ \mathbf{I}_{8} $.
Hence, the combining matrix $ \mathbf{W}^T(m) $ can be rewritten as $ \mathbf{W}^T(m) = \boldsymbol{\Pi} \times \text{diag}\{\mathbf{v}_1(m),\mathbf{v}_2(m)\} $, where $ \mathbf{v}_1(m) = [v_{1,1}(m),v_{1,2}(m),v_{1,3}(m),v_{1,4}(m)]^T $ and $ \mathbf{v}_2(m) = [v_{2,1}(m),v_{2,2}(m),v_{2,3}(m),v_{2,4}(m)]^T $.
\section{Proposed Mixed-field Localization Method}
In this section, we first elaborate the main design challenges when the BS of a hybrid antenna array is used for localizing the mixed-field targets. Then an efficient localization algorithm is proposed to solve them. 

Specifically, compared with existing works on mixed-field localization, the studied problem  differs in the following three main aspects, which introduce new challenges.
First, most existing works on near-field localization method using decoupled subspace methods (e.g., \cite{zhang2018localization}) rely on the assumption that the inter-antenna spacing is no larger than $\frac{\lambda}{4}$, such that reduced-dimensional/modified MUSIC algorithms can be devised to enable near-field localization. However, this assumption is incompatible to legacy wireless communication systems with typically {\emph{half-wavelength inter-antenna spacing}}.
Second, existing angle/location estimation methods for near-field localization primarily focus on digital array systems, which, however, may not be applicable to hybrid antenna systems with reduced-dimensional received signals.\footnote{Although the beamspace MUSIC method was designed to effectively estimate path directions hybrid antenna systems \cite{guo2017millimeter}, it inevitably decreases the maximum number of estimable sources  $ \mathcal{O}(N_{\rm RF})$.} 
Third, many mixed-field localization methods typically differentiate far-field and near-field targets by detecting far-field targets' angles in the angle-domain MUSIC spectrum by letting $ r\to \infty$.
However, it will be shown later in Section~\ref{sec:range estimation} that this method may make false detection of near-field targets, when their ranges are relatively large (still smaller than Rayleigh distance), hence resulting in high localization errors. 
This issue becomes more severe when the half-wavelength inter-antenna spacing constraint is considered due to the appearance of ambiguous targets' angles.

To address the above issues, we propose a new mixed-field localization method as illustrated in Fig. \ref{fig:Algorithm}, which consists of three main steps with their main ideas summarized below and detailed procedures elaborated in  subsequent subsections.  

\begin{itemize}
	\item In Step 1, we accurately recover equivalent received signals at the UPA antennas from the reduced-dimensional signals $\mathbf{y}(m)$, by judiciously designing the analog combining matrices over time $\{\mathbf{W}(m)\}$. 
	In particular, the UPA with half-wavelength inter-antenna spacing is considered, which inevitably causes angular ambiguity and will be addressed in the third step.
	\item In Step 2, we design a modified low-dimensional MUSIC algorithm for UPA-based target localization by constructing a new \emph{Toeplitz-structured} (far-field equivalent) covariance matrix of a virtual S-UPA, which enables effective decoupling of the angle and range estimation.
	\item In Step 3, we propose an efficient target classification method based on distinct characteristics of 3D MUSIC spectrum in the range domain between far-field and near-field targets, which can also effectively address the angular ambiguity issue caused by the virtual S-UPA arsing from the half-wavelength inter-antenna spacing the physical UPA. 
\end{itemize}
\vspace{-10pt}
\subsection{Signal Reconstruction}
\label{sec:signal reconstruction}
Due to the low-dimensional signal matrix in the hybrid antenna systems (i.e., $\mathbf{y}(m)$), the number of targets that can be estimated is significantly reduced. To tackle this issue, we propose an efficient method to recover the equivalent signals received at the UPA and minimize the recovery error by carefully designing the analog combining matrices $\mathbf{W}(m), \forall m \in \mathcal{M}$.

\begin{figure*}[t]
	\begin{equation}
		\label{eq:collected signal}
		\mathbf{Y}_{\ell} \!\!=\!\! \left[\begin{array}{c}
			\mathbf{y}\big(U(\ell-1)+1\big) \\
			\mathbf{y}\big(U(\ell-1)+2\big) \\
			\vdots\\
			\mathbf{y}\big(U(\ell-1)+U\big)
		\end{array}\right]
		\!\!=\!\!
		\left[\begin{array}{c}
			\mathbf{W}\big(U(\ell-1)+1\big) \\
			\mathbf{W}\big(U(\ell-1)+2\big) \\
			\vdots\\
			\mathbf{W}\big(U(\ell-1)+U\big)
		\end{array}\right] \!\!\mathbf{G} \mathbf{s}_{\ell} \!\!+\!\! \left[\begin{array}{c}
			\mathbf{W}\big(U(\ell-1)+1\big) \mathbf{n}\big(U(\ell-1)+1\big) \\
			\mathbf{W}\big(U(\ell-1)+2\big) \mathbf{n}\big(U(\ell-1)+2\big) \\
			\vdots\\
			\mathbf{W}\big(U(\ell-1)+U\big)\mathbf{n}\big(U(\ell-1)+U\big)
		\end{array}\right], \forall \ell \in \mathcal{L}.
	\end{equation}
	\hrulefill
\end{figure*}
To this end, we divide the $ M $ received signals $\mathbf{y}(m)$ into $ L = M/U $ groups, for which the $ \ell $-th group's signals can be concatenated into a new effective signal vector $ \mathbf{Y}_\ell \in \mathbb{C}^{N \times 1} $ as given in \eqref{eq:collected signal}. As such, each group's received pilots can be utilized to recover a sample of  equivalent signals received at the UPA.\footnote{In \eqref{eq:collected signal}, for ease of design, the transmitted signals from targets are assumed to be identical in every $ M $ pilots, i.e., $ \mathbf{s}\big(M(\ell-1)+1\big) = \mathbf{s}\big(M(\ell-1)+2\big) =, \cdots, = \mathbf{s}\big(M(\ell-1)+M\big) \triangleq \mathbf{s}_{\ell} $.}
Moreover, for simplicity, we set $ \mathbf{W}(u) = \mathbf{W}\big(u+\ell U\big), \forall u \in \mathcal{U} \triangleq \{1,2,\cdots,U\}, \forall \ell \in \mathcal{L} $.
As such, if  {\small$ {\mathbf{W}_0} \triangleq [\mathbf{W}^T(1),\mathbf{W}^T(2),\cdots, \mathbf{W}^T(U)]^T $} is an invertible matrix, the $ \ell $-th sample of the reconstructed signals received at the UPA antennas can be recovered as \cite{trigka2018effective}
\begin{equation}
	\label{eq:inverse}
	\hat{\boldsymbol{\eta}}(\ell) = \mathbf{W}_0^{-1}\mathbf{Y}_{\ell} = \mathbf{G} \mathbf{s}_{\ell} + \mathbf{n}_{\ell},
\end{equation}
where the effective noise vector $ \mathbf{n}_{\ell} \in \mathbb{C}^{N \times 1} $ is given by
{\small\begin{equation}
	\label{eq:}
	\mathbf{n}_{\ell} =  \mathbf{W}_0^{-1}\left[\begin{array}{c}
	\mathbf{W}(1) \mathbf{n}\big(U(\ell-1)+1\big) \\
	\mathbf{W}(2) \mathbf{n}\big(U(\ell-1)+2\big) \\
	\vdots\\
	\mathbf{W}(U)\mathbf{n}\big(U(\ell-1)+U\big)
	\end{array}\right].
\end{equation}}
The corresponding signal recovery error for ${\boldsymbol{\eta}}(\ell)$ is  $$ g(\mathbf{W}_0)=|\hat{\boldsymbol{\eta}}(\ell) - \mathbf{G} \mathbf{s}_{\ell}|^2 = {\rm Tr}\big(\mathbf{n}_{\ell}\mathbf{n}_{\ell}^H\big). $$

Then, we formulate an optimization problem to minimize the signal recovery error as follows.
\begin{subequations}
	\label{Eq:optimization problem}
	\begin{align}
		({\bf P1}):~~& \min_{\mathbf{W}_0}~  {\rm Tr}\Big(\mathbf{n}_{\ell}\mathbf{n}_{\ell}^H\Big) \label{eq:object}\\
		&~{\text{s}}{\text{.t}}{\rm{. }} ~~~{\rm Rank}(\mathbf{W}_0) = N,\label{eq:constraint 1}\\
		&~~~~~~~~~|[\mathbf{v}_f(u)]_i| = 1, \forall u, i \in \mathcal{U}, \forall f \in \mathcal{F},\label{eq:constraint 2}
	\end{align}
\end{subequations}
where  constraint \eqref{eq:constraint 1} is enforced to guarantee  $\mathbf{W}_0$ being invertible, and the unit-modulus constraint \eqref{eq:constraint 2} accounts for phase shifters in the analog combining matrix $ \mathbf{W}_0$ and $\{\mathbf{v}_f(u)\}$ are related to $\mathbf{W}_0$ (see \eqref{Eq:W}).

Although Problem ({\bf P1}) is non-convex, the necessary conditions for its optimal solution is given as follows.
\begin{lemma}
	\label{lemma:optimal solution condition}
	\emph{The optimal solution to Problem ({\bf P1}), denoted by $ \mathbf{W}_0^{\ast}$, should satisfy the following two conditions:\\
	\indent 1) $ {\rm Rank}(\mathbf{W}_0^{\ast}) = N $,\\
	\indent 2) $ \sum\limits_{u=1}^U \mathbf{W}^{\ast}(u)^H\mathbf{W}^{\ast}(u) $ is a diagonal matrix.}	
\end{lemma}

Next, we present an optimal solution $ \mathbf{W} $, which satisfy the two conditions in Lemma \ref{lemma:optimal solution condition}.
\begin{lemma}
	\label{lemma:designed W}
	\emph{An optimal solution 
    to ({\bf P1}) can be constructed as $\mathbf{W}_0^*=[(\mathbf{W}^T(1))^*,\cdots,(\mathbf{W}^T(u))^*,\cdots, (\mathbf{W}^T(U)]^T)^*$, where $(\mathbf{W}^T(u))^*$ is given by
	\begin{equation}
		\label{eq:design of Wp} 
		(\mathbf{W}^T(u))^* = \boldsymbol{\Pi} \!\times\! \left[\begin{array}{cccc}
		\mathbf{f}_{u} & \mathbf{0} & \ldots & \mathbf{0} \\
		\mathbf{0} & \mathbf{f}_{u} & \ldots & \mathbf{0} \\
		\vdots & \vdots & \ddots & \vdots \\
		\mathbf{0} & \mathbf{0} & \ldots & \mathbf{f}_{u}
		\end{array}\right].
	\end{equation}
    Herein, $\mathbf{f}_u$ is the $ u $-th column of the DFT matrix $ \mathbf{F} = [\mathbf{f}_1, \cdots, \mathbf{f}_u, \cdots, \mathbf{f}_U] \in \mathbb{C}^{U \times U} $ with $\mathbf{f}_u$ given by
		\begin{equation}
			\mathbf{f}_u = \Big[1, e^{\jmath 2\pi \frac{u-1}{U}}, e^{\jmath 2\pi \frac{2(u-1)}{U}}, \cdots, e^{\jmath 2\pi \frac{(U-1)(u-1)}{U}}\Big]^H.
		\end{equation}	
	}
\end{lemma}
\vspace{-0.2cm}
\begin{proof}
	Based on \eqref{eq:design of Wp}, $ \mathbf{W}_0^{H}\mathbf{W}_0 = \sum\limits_{u=1}^U \mathbf{W}^H(u)\mathbf{W}(u) = \boldsymbol{\Pi} \times \mathbf{F} \times\boldsymbol{\Pi}^H$.
Since {$\small\mathbf{F} = \text{diag}\Big\{{\sum\limits_{u=1}^U\mathbf{f}_{u}\mathbf{f}_{u}^H,\cdots,\sum\limits_{u=1}^U\mathbf{f}_{u}\mathbf{f}_{u}^H}\Big\}, $} and $ \sum\limits_{u=1}^U\mathbf{f}_{u}\mathbf{f}_{u}^H  = U\mathbf{I}_U$, we have $ \mathbf{W}_0^{H}\mathbf{W}_0 = U\mathbf{I}_N $, which indicates that $ {\mathbf{W}}_0/{\sqrt{U}} $ is an orthogonal matrix.
	As such, we have $ \rank(\mathbf{W}_0) = N $ and the designed combining matrix $ \mathbf{W}_0 $ satisfy the two conditions in Lemma \ref{lemma:optimal solution condition}, thus  completing the proof. 
\end{proof}

From Lemma \ref{lemma:designed W}, we can recover the equivalent signals received at the UPA without amplifying the noise power, which will be used for target localization in the sequel.   
\vspace{-8pt}
\subsection{Angle Estimation}
\label{sec:estimation for far-field users}
In this subsection, we develop a decoupled MUSIC algorithm tailored to the UPA architecture. 
Owing to the half-wavelength spacing constraint, it is difficult to directly construct an ambiguity-free equivalent far-field covariance matrix.
To address this issue, we first construct a far-field covariance matrix of a virtual sparse UPA (S-UPA) from the original covariance matrix, which enables the decoupling of the joint 3D angle-and-range estimation into separate angle and range estimations.
Based on the resulting angle-only covariance matrix of the virtual S-UPA, a decoupled 2D angular MUSIC algorithm is further devised to separately estimate the elevation and azimuth angles.
The sparse structure of the virtual S-UPA inevitably introduces angular ambiguities, which will be effectively resolved in Section~\ref{sec:range estimation}.

\underline{\textbf{Equivalent far-field covariance matrix:}}
To decouple the angle and range estimation, we first obtain the covariance matrix of the signals received at the UPA in \eqref{eq:inverse}, given by
\begin{equation}
\label{eq:orignal covariance matrix}
	\mathbf{R}_{\boldsymbol{\eta}} = \mathbb{E}\!\left\{\hat{\boldsymbol{\eta}}(\ell) (\hat{\boldsymbol{\eta}}(\ell))^{\mathrm{H}}\right\} = \mathbf{A}\mathbf{S}_{\rm F}\mathbf{A}^{{H}} + \mathbf{B}\mathbf{S}_{\rm N}\mathbf{B}^{{H}} + \sigma^2 \mathbf{I}_{N},
\end{equation}
where $ \mathbf{S}_{\rm F} = \mathbb{E}\left\{\mathbf{s}_{\rm F}(t) \mathbf{s}_{\rm F}^{\mathrm{H}}(t)\right\} = \text{diag}\{g_1,g_2,\cdots,g_{K_1}\} $ and $ \mathbf{S}_{\rm N} = \mathbb{E}\left\{\mathbf{s}_{\rm N}(t) \mathbf{s}_{\rm N}^{\mathrm{H}}(t)\right\} = \text{diag}\{g_{K_1 + 1},g_{K_1+2},\cdots,g_{K}\} $ represent the source covariance matrices for the far- and near-field targets.
In practice, the theoretical $ \mathbf{R} $ is estimated by the sample average covariance matrix $ \hat{\mathbf{R}} $, given by
\begin{equation}
	\hat{\mathbf{R}}_{\boldsymbol{\eta}}=\frac{1}{L} \sum_{\ell=1}^L \hat{\boldsymbol{\eta}}(\ell) \hat{\boldsymbol{\eta}}(\ell)^{{H}},
\end{equation}
where $ \ell \in \mathcal{L} $ denotes the effective received pilot index with $ L $ denoting the number of pilots.
Based on \eqref{eq:received signal}, the $ (i,j) $-th entry of the covariance matrix $\mathbf{R}_{\boldsymbol{\eta}} $ is given by 
\begin{equation}
\begin{aligned}
\!\!\!\!\!\![\mathbf{R}_{\boldsymbol{\eta}}]_{i, j} &= \sum_{k=1}^K \! g_k \exp\left({-\jmath\frac{2\pi}{\lambda}\Big(u_i(\alpha_{k},\beta_{k})-u_j(\alpha_{k},\beta_{k})}\right)\\
&~~~{+v_i(r_{k},\alpha_{k},\beta_{k})-v_j(r_{k},\alpha_{k},\beta_{k})\Big)} \!\!+\!\! \sigma^2 \delta_{i, j},
\end{aligned}	
\end{equation}
where $ \delta_{i, j}=1 $  when $i=j$, and $ \delta_{i, j}=0$  otherwise. 
It is observed that the anti-diagonal entries of $ \mathbf{R} $ only involve the angle information of both far- and near-field targets, which is given by
\begin{equation}
\label{eq:anti-diagonal}
\begin{aligned}
[\mathbf{R}_{\boldsymbol{\eta}}]_{n, N+1-n }\!=\!\! \sum_{k=1}^K \! g_k e^{j \frac{2\pi}{\lambda}2d_0(n_x\alpha_k + n_y\beta_k)} \!\!+\!\! \sigma^2 \delta_{n, N \!+\!1\!-\! n}, n\in \mathcal{N},
\end{aligned}
\end{equation}
which endows the possibility to estimate the angles of targets.

To this end, we  reorganize \eqref{eq:anti-diagonal} into a new vector $ {\tilde{\boldsymbol{\eta}}} \in \mathbb{C}^{N \times 1} $ without noise taken into account, which is given by
\begin{equation}
	\label{eq:anti}
	[{\tilde{\boldsymbol{\eta}}}]_n = \sum_{k=1}^K g_k e^{j \frac{2\pi}{\lambda}2d_0(n_x\alpha_k + n_y\beta_k)}, \forall n \in \mathcal{N}.
\end{equation}
Then, we shall show that the entries of ${\tilde{\boldsymbol{\eta}}}$ can be used to construct a new matrix, which is equivalent to the covariance matrix of a virtual S-UPA, hence allowing for applying  
subspace-based angle estimation methods. 

Specifically, we consider an S-UPA with
$\widetilde{N} \!\!=\! \widetilde{N}_x\widetilde{N}_y $ antennas where  $ \widetilde{N}_x \!=\! \frac{N_x+1}{2} $ and $ \widetilde{N}_y \!=\! \frac{N_y+1}{2} $, and the inter-antenna spacing is  $\lambda$. Then, 
 the received signas at the virtual S-UPA is given by $$ 	{\tilde{\boldsymbol{\eta}}}_{\rm SUPA}(\ell) = \tilde{\mathbf{V}}_{\rm SUPA} \mathbf{s}(\ell) + \mathbf{n}_{\rm SUPA}(\ell),$$ where $ \widetilde{\mathbf{V}}_{\rm SUPA} = [\tilde{\mathbf{a}}_{\rm SUPA}\left(\theta_{1},\phi_1\right),\cdots,\tilde{\mathbf{a}}_{\rm SULA}\left(\theta_{K},\phi_K\right)] $. Herein, we define $\tilde{\mathbf{a}}_{\rm SULA}\left(\theta_{k},\phi_K\right) =   \tilde{\mathbf{a}}_y(\alpha_{{k}})\otimes \tilde{\mathbf{a}}_x(\beta_{{k}})$, 
where $$ \tilde{\mathbf{a}}_x^{H}(\alpha_{{k}}) = [e^{j\frac{4\pi}{\lambda}(-\frac{\widetilde{N}_x-1}{2})\alpha_k},\cdots, e^{j\frac{4\pi}{\lambda}\frac{\widetilde{N}_x-1}{2}\alpha_k}] \in \mathbb{C}^{\widetilde{N}_x\times1},$$   $$ \tilde{\mathbf{a}}_y^{H}(\beta_{{k}})=[e^{j\frac{4\pi}{\lambda}(-\frac{\widetilde{N}_y-1}{2})\beta_k},\cdots, e^{j\frac{4\pi}{\lambda}\frac{\widetilde{N}_y-1}{2}\beta_k}] \in \mathbb{C}^{\tilde{N_y}\times1}$$ represent the channel response vectors of the equivalent S-UPA in the $ x $-axis and $ y $-axis, respectively.
Hence, the covariance matrix of this S-UPA can be obtained as $ \widetilde{\mathbf{R}}_{{\tilde{\boldsymbol{\eta}}}} = \mathbb{E}\!\left\{{\tilde{\boldsymbol{\eta}}}_{\rm SULA}(\ell) ({\tilde{\boldsymbol{\eta}}}_{\rm SULA}(\ell))^{{H}}\right\} $, given by
{\small
	\begin{equation}
	\widetilde{\mathbf{R}}_{{\tilde{\boldsymbol{\eta}}}} =\left[\begin{array}{cccccc}
	\widetilde{\mathbf{R}}_{1,1} & \widetilde{\mathbf{R}}_{1,2} & \cdots & \widetilde{\mathbf{R}}_{1,\widetilde{N}_y} \\
	\widetilde{\mathbf{R}}_{2,1} & \widetilde{\mathbf{R}}_{2,2} & \cdots & \widetilde{\mathbf{R}}_{2,\widetilde{N}_y} \\
	\vdots & \vdots & \cdots & \vdots \\
	\widetilde{\mathbf{R}}_{\widetilde{N}_y,1} & \widetilde{\mathbf{R}}_{\widetilde{N}_y,2} & \cdots & \widetilde{\mathbf{R}}_{\widetilde{N}_y,\widetilde{N}_y}
	\end{array}\right],
	\end{equation}}where $ \widetilde{\mathbf{R}}_{i,j}, \forall i,j \in \widetilde{\mathcal{N}}_y \triangleq \{1,2,\cdots, \widetilde{N}_y\} $ is a Toeplitz matrix. Specifically, the ($i_x$, $j_x$)-th entry in $\widetilde{\mathbf{R}}_{i,j}$ is given by
    \begin{equation}
    \label{eq:R_ij}
        [\widetilde{\mathbf{R}}_{i,j}]_{i_x,j_x} = \sum\limits_{k=1}^{K}e^{\jmath\frac{4\pi}{\lambda}(i-j)\beta_k}e^{\jmath\frac{4\pi}{\lambda}(i_x-j_x)\alpha_k},
    \end{equation}
    where $ i_x,j_x \in \widetilde{\mathcal{N}}_x \triangleq \{1,2,\cdots, \widetilde{N}_x\}$.
\begin{figure}[t]	
	\vspace{-14pt}
	\centering
	\subfloat[Parameter $ \beta $ estimation]{
		\includegraphics[width=0.25\textwidth]{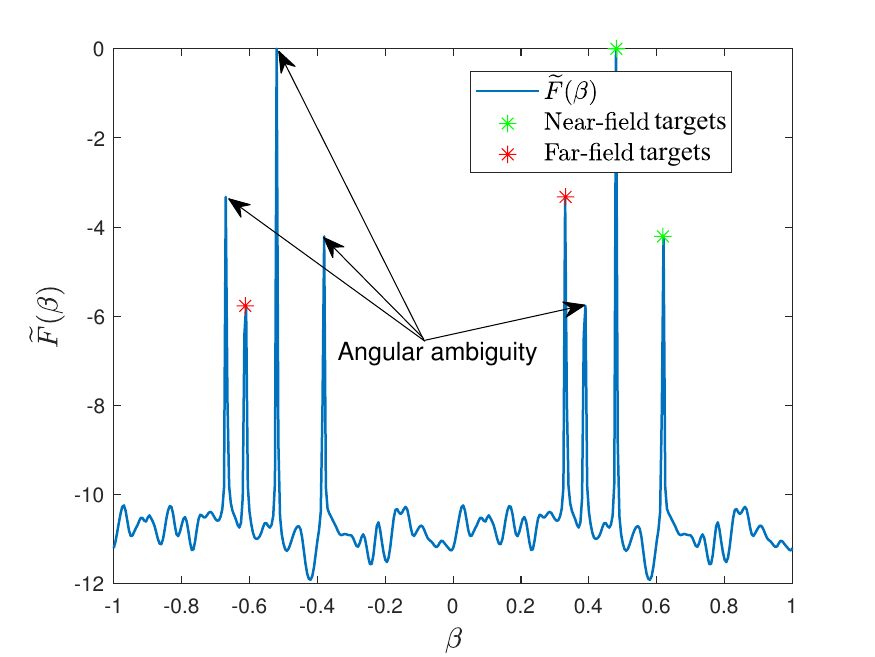}}
	\subfloat[Parameter $ \alpha $ estimation]{
		\includegraphics[width=0.25\textwidth]{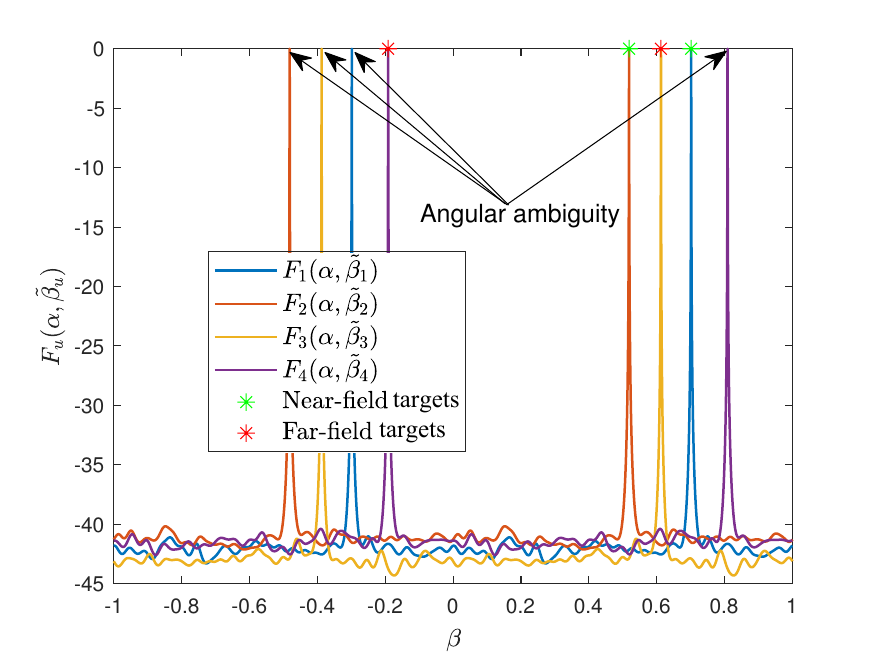}}
	\caption{Illustration of angle estimation. System parameters are set as $f = 10$ GHz and $N_x = N_y = 61$. The elevation and azimuth angles for far-field targets are ($\frac{\pi}{3}$, $-\frac{\pi}{4}$, $1000$ m) and ($\frac{\pi}{8}$, $\frac{2\pi}{3}$, $1500$ m), while the locations for near-field targets are set as ($\frac{5\pi}{13}$, ${0.23\pi}$, $50$ m) and ($\frac{\pi}{4}$, $\frac{5\pi}{21}$, $60$ m).}
	\label{fig:near-field UAV}
	\vspace{-8pt}
\end{figure}

By comparing \eqref{eq:anti} and  $\widetilde{\mathbf{R}}_{i,j} $ in \eqref{eq:R_ij}, we can readily obtain the following result.
\begin{lemma}\label{Lem:Gene}
	\emph{The entries of ${\tilde{\boldsymbol{\eta}}}$ in  \eqref{eq:anti} can be used to construct the matrix $\widetilde{\mathbf{R}}$. Mathematically, each entry of $[\widetilde{\mathbf{R}}]_{i,j}$ is given  by}
	\begin{equation}
	\begin{aligned}
		\label{eq:toeplitz}
		\widetilde{\mathbf{R}}_{i,j} &\!=\! {\rm{toeplitz}}\big([\tilde{\mathbf{y}}]_{\frac{N+1}{2}-(i-1)N_x:\frac{N+2-N_x}{2}\!-\!(i-1)N_x},\big.\\
		 &~~~~~~~~~~\big.[\tilde{\mathbf{y}}]_{\frac{N+1}{2}\!+\!(j-1)N_x 			\!:\!\frac{N+N_x}{2}\!+\!(j-1)N_x}\big), i,j \in \widetilde{\mathcal{N}}_y.
	\end{aligned}
	\end{equation}
\end{lemma}

Given Lemma~\ref{Lem:Gene}, the 2D-angular MUSIC algorithm can be utilized to estimate targets' angles by leveraging eigenvectors of the covariance matrix $ \widetilde{\mathbf{R}}_{{\tilde{\boldsymbol{\eta}}}} $, given by
\begin{equation}
	\label{eq:near-field EVD}
	\widetilde{\mathbf{R}}_{{\tilde{\boldsymbol{\eta}}}} = \widetilde{\mathbf{U}}_s \widetilde{\boldsymbol{\Sigma}}_s \widetilde{\mathbf{U}}_s^{{H}}+\widetilde{\mathbf{U}}_n \widetilde{\boldsymbol{\Sigma}}_n \widetilde{\mathbf{U}}_n^{{H}}, 
\end{equation}
where $ \widetilde{\mathbf{U}}_s \in \mathbb{C}^{\widetilde{N}\times K}$ and $ \widetilde{\mathbf{U}}_n \in \mathbb{C}^{\widetilde{N}\times \widetilde{N}-K}$ denote the signal and noise subspaces, respectively.
In addition, $ \widetilde{\boldsymbol{\Sigma}}_s $ and $ \widetilde{\boldsymbol{\Sigma}}_n $ are diagonal matrices, representing the corresponding eigenvalues of the signal and noise subspaces, respectively.
To be specific, based on the 2D MUSIC principle, the azimuth and elevation angles of both far-field and near-field targets can be estimated by searching peaks in the 2D-angular MUSIC spectrum \cite{he2011efficient}
\begin{equation}
	\label{eq:MUSIC for far-field}
	F(\alpha, \beta) = \frac{1}{\tilde{\mathbf{a}}^{H}_{\rm SUPA}(\alpha, \beta){\widetilde{\mathbf{U}}}_n {\widetilde{\mathbf{U}}}_n^{{H}} \tilde{\mathbf{a}}^{H}_{\rm SUPA}(\alpha, \beta)}.
\end{equation} 

However, it is worth noting that this method introduces two issues in target angle estimation. First, there exist angular ambiguity in the angle estimation in \eqref{eq:MUSIC for far-field}, due to the sparsity of the virtual S-UPA, i.e., $\tilde{\mathbf{a}}^{H}_{\rm SUPA}(\alpha, \beta) = \tilde{\mathbf{a}}^{H}_{\rm SUPA}(\alpha \pm 1, \beta \pm 1) $.
Second, performing a grid search in the 2D-angular domain incurs prohibitively high computational complexity. These issues will be addressed in the sequel.

\underline{\textbf{Decoupled 2D-angular MUSIC algorithm:}}
To reduce  the computational complexity in the 2D-angular domain search for $\{\alpha, \beta\}$, we propose a decoupled 2D-angular MUSIC algorithm, where the parameters $\beta$ and $\alpha$ are sequentially estimated by a 1D grid search. 

To this end, we first rewrite the denominator of $F(\alpha,\beta) $ in $ \eqref{eq:MUSIC for far-field} $ as
\begin{align}
	f(\alpha, \beta) &= \tilde{\mathbf{a}}^{H}_{\rm SUPA}(\alpha, \beta){\widetilde{\mathbf{U}}}_n {\widetilde{\mathbf{U}}}_n^H \tilde{\mathbf{a}}_{\rm SUPA}(\alpha, \beta)\notag\\
	&=\big(\tilde{\mathbf{a}}_y(\beta)\otimes\tilde{\mathbf{a}}_x(\alpha)\big)^H{\widetilde{\mathbf{U}}}_n {\widetilde{\mathbf{U}}}_n^H 	\big(\tilde{\mathbf{a}}_y(\beta)\otimes\tilde{\mathbf{a}}_x(\alpha)\big)\notag\\
	&=\tilde{\mathbf{a}}_x^H(\alpha)\widetilde{\mathbf{T}}(\beta)\tilde{\mathbf{a}}_x(\alpha),
\end{align}
where $ \widetilde{\mathbf{T}}(\beta) = \big(\tilde{\mathbf{a}}_y(\beta)\otimes\mathbf{I}_{N_x}\big)^H{\widetilde{\mathbf{U}}}_n{{\widetilde{\mathbf{U}}}}_n^H \big(\tilde{\mathbf{a}}_y(\beta)\otimes\mathbf{I}_{N_x}\big)$.
Then, the estimated angles can be obtained by solving the following problem by leveraging the MUSIC principle
\begin{equation}
\begin{array}{ll}
(\textbf{\text{P2}})&\min\limits_{\alpha, \beta}~~  \tilde{\mathbf{a}}_x^H(\alpha)\widetilde{\mathbf{T}}(\beta)\tilde{\mathbf{a}}_x(\alpha) \\
&~\text {s.t.} ~~~\tilde{\mathbf{c}}_{{\widetilde{N}_x}}^{{H}} \tilde{\mathbf{a}}_x(\alpha) = 1,
\end{array}
\end{equation}
where $ \tilde{\mathbf{c}}_{{\widetilde{N}_x}}^T = [\underbrace{0,\cdots,0,}_{\frac{\widetilde{N}_x-1}{2}}1,\underbrace{0,\cdots,0}_{\frac{\widetilde{N}_x-1}{2}}] $ and the constraint of $ (\textbf{\text{P2}}) $ is to avoid the trivial solution $ \tilde{\mathbf{a}}_x(\alpha) = \mathbf{0} $.
Although it is not easy to obtain the optimal solution to Problem (\textbf{\text{P2}}) in closed form, we present its poperies in the lemma below.
\begin{lemma}
	\label{lemma:solution to P2}
	\emph{The optimal solution to (\textbf{\text{P2}}) should satisfy }
	\begin{equation}
	\label{eq:optimal beta}
	\beta^\ast = \arg\max_{\beta}\tilde{\mathbf{c}}_{{\widetilde{N}_x}}^{{H}}\widetilde{\mathbf{T}}(\beta)^{-1}\tilde{\mathbf{c}}_{{\widetilde{N}_x}},
	\end{equation}
    \begin{equation}
    \label{eq:phase for alpha}
        \tilde{\mathbf{a}}_x(\alpha^\ast) = \frac{\widetilde{\mathbf{T}}(\beta^\ast)^{-1} \tilde{\mathbf{c}}_{{\widetilde{N}_x}}}{\tilde{\mathbf{c}}_{{\widetilde{N}_x}}^{{H}} \widetilde{\mathbf{T}}(\beta^\ast)^{-1}\tilde{\mathbf{c}}_{{\widetilde{N}_x}}}.
    \end{equation}
\end{lemma}
\begin{proof}
	We first construct a Lagrange function $ L(\alpha, \beta) $ for $ (\textbf{\text{P2}}) $, which is given by
	\begin{equation}
		L(\alpha, \beta)= \tilde{\mathbf{a}}_x(\alpha)^{H} \widetilde{\mathbf{T}}(\beta) \tilde{\mathbf{a}}_x(\alpha) - \lambda\left(\tilde{\mathbf{c}}_{{\widetilde{N}_x}}^{{H}} \tilde{\mathbf{a}}_x(\alpha)-1\right),
	\end{equation}
	where $ \lambda $ denotes the Lagrange multiplier. 
	The partial derivative of $ L(\alpha, \beta) $ with respect to $ \tilde{\mathbf{a}}_x(\alpha) $ is $ \frac{\partial L(\alpha, \beta)}{\partial \tilde{\mathbf{a}}_x(\alpha)} = 2 \widetilde{\mathbf{T}}(\beta)\mathbf{a}_x(\alpha) + \lambda \mathbf{c}_{{N_x}} $. To obtain the minimal value, we set $ \frac{\partial L(\alpha, \beta)}{\partial \tilde{\mathbf{a}}_x(\alpha)} = 0 $, which can be simplified as $ \tilde{\mathbf{a}}_x(\alpha^\ast) = \lambda\widetilde{\mathbf{T}}(\beta)^{-1}\tilde{\mathbf{c}}_{{\widetilde{N}_x}} $. Considering that $ \tilde{\mathbf{c}}_{{\widetilde{N}_x}}^{{H}} \tilde{\mathbf{a}}_x(\alpha^\ast) = 1 $, $ \tilde{\mathbf{a}}_x(\alpha^\ast) $ can be rewritten as
    $\tilde{\mathbf{a}}_x(\alpha^\ast) = \frac{\widetilde{\mathbf{T}}(\beta)^{-1} \tilde{\mathbf{c}}_{{\widetilde{N}_x}}}{\tilde{\mathbf{c}}_{{\widetilde{N}_x}}^{{H}} \widetilde{\mathbf{T}}(\beta)^{-1}\tilde{\mathbf{c}}_{{\widetilde{N}_x}}}$.
    By substituting $\tilde{\mathbf{a}}_x(\alpha^\ast)$ into $ f(\alpha, \beta) $, we obtain $ f(\alpha^\ast, \beta) = \frac{1}{\tilde{\mathbf{c}}_{{\widetilde{N}_x}}^{{H}}\widetilde{\mathbf{T}}(\beta)^{-1}\tilde{\mathbf{c}}_{{\widetilde{N}_x}}} $.
    Hence, the optimal solution $\beta^\ast$ can be expressed as $\beta^\ast = \arg\max\limits_{\beta}\tilde{\mathbf{c}}_{{\widetilde{N}_x}}^{{H}}\widetilde{\mathbf{T}}(\beta)^{-1}\tilde{\mathbf{c}}_{{\widetilde{N}_x}}$, thus completing the proof.
\end{proof}

Lemma~\ref{lemma:solution to P2} shows that for the joint estimation of $\{\alpha, \beta\}$, we can first estimate $\beta$ according to \eqref{eq:optimal beta}, followed by the estimation for $\alpha$.
\subsubsection{Estimation of angle parameter $\beta$} 
First, to estimate $\beta$, it can be easily shown from \eqref{eq:optimal beta} that $\beta$ can be estimated by searching spectrum peaks in the following spectrum function
\begin{equation}
\label{eq:far-field beta}
    \widetilde{F}(\beta) = \tilde{\mathbf{c}}_{{\widetilde{N}_x}}^{{H}}\widetilde{\mathbf{T}}(\beta)^{-1}\tilde{\mathbf{c}}_{{\widetilde{N}_x}}.
\end{equation}
For target $k\in \mathcal{K}$, let $\hat{\beta}_k$ denote its estimated $\beta$.
However, since $ \tilde{\mathbf{a}}_y(\beta) = \tilde{\mathbf{a}}_y(\beta \pm 1) $ and hence $ \widetilde{\mathbf{T}}(\beta) = \widetilde{\mathbf{T}}(\beta\pm 1) $, we have $\widetilde{F}(\hat{\beta}_k) = \widetilde{F}(\hat{\beta}_k \pm 1) $.
Therefore, for each estimated true parameter $\hat{\beta}_k$, there always exists one ambiguous parameter $\hat{\beta}_k^{\rm (am)}$, which is given by
\begin{equation}
    \label{eq:angular ambiguity}
    \hat{\beta}_k^{(\rm am)} =\left\{\begin{matrix}
    \hat{\beta}_k + 1,~~\hat{\beta}_k = \sin\theta_k\cos\phi_k < 0,\\ 
    \hat{\beta}_k - 1,~~\hat{\beta}_k =  \sin\theta_k\cos\phi_k > 0.
    \end{matrix}\right.
\end{equation}

As such, for $K$ targets, a total of  $2K$ peaks can be found in the spectrum function $\widetilde{F}(\beta)$ including $K$ true angles and $K$ corresponding ambiguous angles, as illustrated in Fig.~\ref{fig:near-field UAV}(a).
For convenience, the estimated parameters are collectively denoted as $ \tilde{\beta}_u \in \{\hat{\beta}_k^{(\rm am)}, \hat{\beta}_k\}, u=1,2,\cdots,2K $. 
\subsubsection{Estimation for angle parameter $\alpha$}
Next, for parameter $\alpha$, instead of estimating it based on \eqref{eq:phase for alpha}, we propose to perform a 1D peak search on the spectrum function $F(\alpha,\beta)$ in \eqref{eq:MUSIC for far-field}, which is generally more accurate owing to the super-resolution of subspace-based methods. Specifically, for each estimated $\tilde{\beta}_{u}$,  $\alpha$ is found by a 1D peak search for the original spectrum function in \eqref{eq:MUSIC for far-field} 
{\small
	\begin{equation}
	\label{eq:far-field alpha estimation}
		F_{u}(\alpha,\!\tilde{\beta}_{u}) \!\!=\!\! \frac{1}{\big(\tilde{\mathbf{a}}_y(\tilde{\beta}_{u})\!\otimes\!\tilde{\mathbf{a}}_x(\alpha)\big)^H\!{\widetilde{\mathbf{U}}}_n {\widetilde{\mathbf{U}}}_n^{{H}}\! \big(\tilde{\mathbf{a}}_y(\tilde{\beta}_{u})\!\otimes\!\tilde{\mathbf{a}}_x(\alpha)\big)}.
	\end{equation}
}Similarly, for each estimated $\tilde{\beta}_{u}$, two peaks will be found in the spectrum (as illustrated in Fig. \ref{fig:near-field UAV}(b)), corresponding to the true and ambiguous ones, denoted by $\hat{\alpha}_k$ and $\hat{\alpha}_k^{{(\rm am)}}$, respectively, for which
\begin{equation}
	\label{eq:angular ambiguity alpha}
	\hat{\alpha}_k^{{(\rm am)}} =\left\{\begin{matrix}
	\hat{\alpha}_k + 1,~~\hat{\alpha}_k = \sin\theta_k\sin\phi_k < 0,\\ 
	\hat{\alpha}_k - 1,~~\hat{\alpha}_k =  \sin\theta_k\sin\phi_k > 0.
	\end{matrix}\right.
\end{equation}
This can be \emph{intuitively} understood, since by leveraging the \emph{sparsity} of  virtual S-UPA, we have $ \tilde{\mathbf{a}}_x(\alpha) = \tilde{\mathbf{a}}_x(\alpha \pm 1) $ and hence $ F_{u}(\alpha \pm 1, \tilde{\beta}_{u}) = F_{u}(\alpha, \tilde{\beta}_{u})$. 
Moreover, it can be shown that for both $ \hat{\beta}_k $ and $ \hat{\beta}_k^{(\rm am)} $, the corresponding sets of estimated $ \{\hat{\alpha}_k, \hat{\alpha}_k^{{(\rm am)}}\}$ are the same. 
\begin{figure}[t]
	\vspace{-14pt}
	\includegraphics[width=0.425\textwidth]{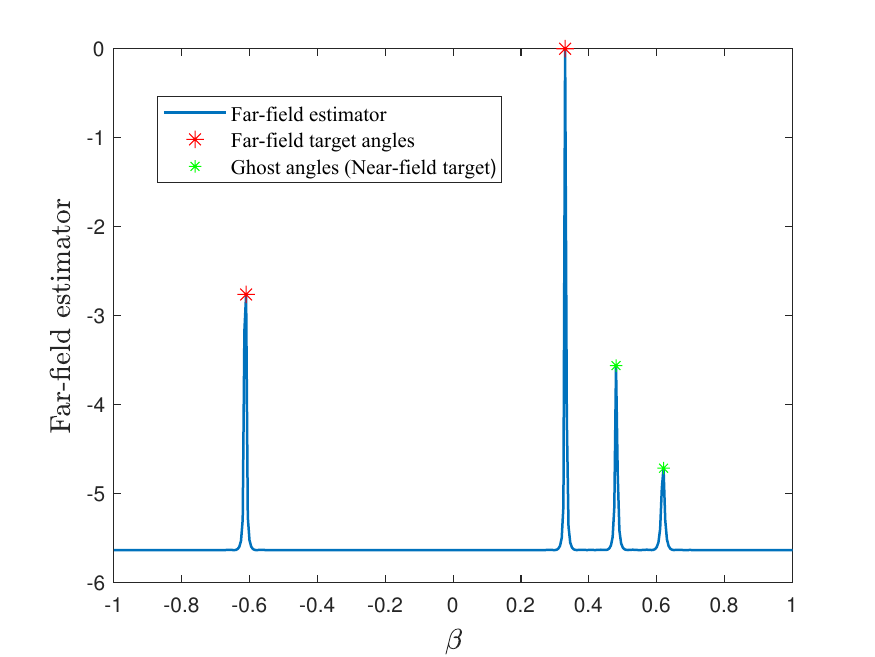}
	\centering
	\caption{Spectrum of far-field estimator in \eqref{eq:2D MUSIC for far-field} proposed in \cite{he2011efficient}. System parameters are set as $f = 10$ GHz and $N_x = N_y = 61$. The elevation and azimuth angles for far-field targets are ($\frac{\pi}{3}$, $-\frac{\pi}{4}$, $1000$ m) and ($\frac{\pi}{8}$, $\frac{2\pi}{3}$, $1500$ m), while the locations for near-field targets are set as ($\frac{5\pi}{13}$, ${0.23\pi}$, $30$ m) and ($\frac{\pi}{4}$, $\frac{5\pi}{21}$, $40$ m).}
	\label{fig:incorrectClassification}
	\vspace{-12pt}
\end{figure}
Based on the above, a total of $ 4K $ pairs of parameters ($\alpha$, $\beta$) can be estimated by the proposed method, which are collected into the following set 
\begin{equation}
	\begin{aligned}
		\mathcal{A} \triangleq \{(\hat{\alpha}_k, \hat{\beta}_k),&(\hat{\alpha}_k^{(\rm am)}, \hat{\beta}_k),\\ &(\hat{\alpha}_k, \hat{\beta}_k^{(\rm am)}),  (\hat{\alpha}_k^{(\rm am)}, \hat{\beta}_k^{(\rm am)}), \forall k \in \mathcal{K}\}.
	\end{aligned}
\end{equation} 
For each parameter pair $(\alpha_g, \beta_g) \in \mathcal{A}, g \in \mathcal{G} \triangleq\{1,2,\cdots,4K\}$, by solving the equations $ \sin\theta_g\cos\phi_g = {\alpha}_{g} $ and $ \sin\theta_{g}\sin\phi_{g} = {\beta}_{g} $, the candidate angle of the target can be obtained as 
\begin{equation}
	{\theta}_{g} = \arcsin\sqrt{\hat{\alpha}_{g}^2+\hat{\beta}_{g}^2},
	\vspace{-5pt}
\end{equation}
\begin{equation}
	{\phi}_{g} = \text{angle}({\hat{\alpha}_{g}+\jmath\hat{\beta}_{g}}),
\end{equation}
where $ {\theta}_{g} $ and $ {\phi}_{g} $ denote  candidate elevation and azimuth angles of the target.
As such, we can obtain $4K$ candidate 2D angles including the true and ambiguous ones for the far-field and near-field targets. 
Although angular ambiguity cannot be removed in the angle estimation,  an effective classification method will be presented in the next to distinguish far-field and near-field targets, as well as eliminate angular ambiguity.

\vspace{-4pt}
\subsection{Target Classification, Near-field Target Range Estimation and Ambiguity Elimination}
\label{sec:range estimation}
In this section, given the candidate  2D-angular parameters of targets, we jointly classify the mixed-field targets, resolve angular ambiguity, and estimate the ranges of targets, by capitalizing on distinct spectrum patterns of far-field and near-field targets in the range domain.

Specifically, by the eigenvalue decomposition (EVD) of $ \mathbf{R}_{\boldsymbol{\eta}} $ in \eqref{eq:orignal covariance matrix}, the original covariance matrix can be decomposed as
\begin{equation}
	{{\mathbf{R}}}_{\boldsymbol{\eta}} = {\mathbf{U}}_s {\boldsymbol{\Sigma}}_s {\mathbf{U}}_s^{{H}}+{\mathbf{U}}_n {\boldsymbol{\Sigma}}_n {\mathbf{U}}_n^{{H}}, 
\end{equation}
where $ {\mathbf{U}}_s \in \mathbb{C}^{N \times K} $ and $ {\mathbf{U}}_n \in \mathbb{C}^{N \times (N-K)} $ denote the signal and noise subspaces, respectively. Moreover, $ {\boldsymbol{\Sigma}}_s \in \mathbb{C}^{K \times K} $ is a diagonal matrix consisting of $ K $ largest eigenvalues, while $ {\boldsymbol{\Sigma}}_n \in \mathbb{C}^{(N-K) \times (N-K)} $ is a diagonal matrix containing $ (N-K) $ smallest eigenvalues.
Before introducing our proposed classification method, we first discuss the issues in existing mixed-field target classification methods  \cite{he2011efficient,liu2013efficient}.

\underline{\textbf{Issues of existing classification method:}} 
In \cite{he2011efficient,liu2013efficient}, the authors propose to obtain the far-field target parameter $\beta$ from the  spectrum $F(r, \theta, \phi)$ by letting $r\to \infty$
\begin{equation}
	\label{eq:2D MUSIC for far-field}
	{F}(\infty, \theta, \phi) = \frac{1}{\mathbf{b}^{H}(\infty,\theta, \phi){\mathbf{U}}_n {\mathbf{U}}_n^{{H}} \mathbf{b}(\infty, \theta, \phi)}.
\end{equation} 
Since the peaks of far-field targets in 3D MUSIC spectrum appears at $r\to \infty$ and $ \theta=\theta_{k}, \phi=\phi_{k}, k\in\mathcal{K}_1 $ \cite{he2011efficient}, the estimator in \eqref{eq:1D far-field beta} can exactly estimate $K_1$ far-field targets.
Similar to Section \ref{sec:estimation for far-field users}, the 2D-angular MUSIC estimator in \eqref{eq:2D MUSIC for far-field} can be decoupled by two 1D MUSIC spectrums.
For simplicity, we only provide the estimator for parameter $\beta$, given by
\begin{equation}
	\label{eq:1D far-field beta}
	\bar{F}(\beta) = \mathbf{c}_{{N_x}}^{{H}}\mathbf{T}(\beta)^{-1}\mathbf{c}_{{N_x}},
\end{equation}
where $ \mathbf{T}(\beta) = \big(\mathbf{a}_y(\beta)\otimes\mathbf{I}_{N_x}\big)^H{\mathbf{U}}_n{\mathbf{U}}_n^H \big(\mathbf{a}_y(\beta)\otimes\mathbf{I}_{N_x}\big)$ and $ \mathbf{c}_{{N_x}}^T =[\underbrace{0,\cdots,0,}_{\frac{N_x-1}{2}}1,\underbrace{0,\cdots,0}_{\frac{N_x-1}{2}}] $.

\begin{figure}[t]
	\vspace{-14pt}
	\includegraphics[width=0.425\textwidth]{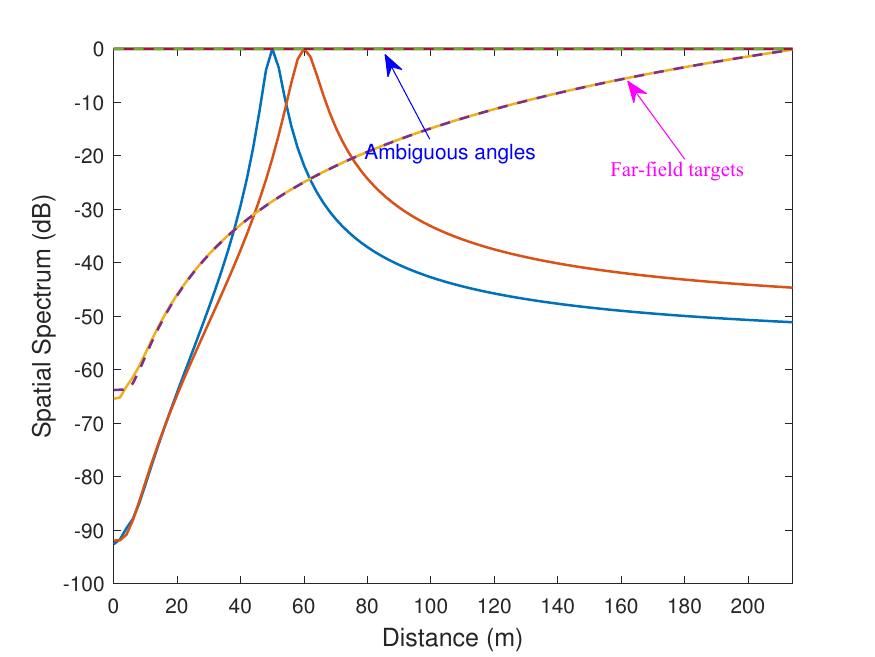}
	\centering
	\caption{Illustration of proposed classification method. System parameters are set as $f = 10$ GHz and $N_x = N_y = 61$. The ranges of the two near-field targets are $50$ and $60$ m.}
	\label{fig:ProposedClasification}
	\vspace{-16pt}
\end{figure}
Note that, when near-field targets approach the Rayleigh distance, there will exit peaks associated with near-field targets in the spectrum $\bar{F}(\beta)$, which results in mistake in the far-field versus near-field target classification.
This can be intuitively understood, since in this case, the near-field channel is highly correlated to the far-field channel at the same angle,
thus making it possible to be wrongly detected as far-field target. 
For illustration, we plot the spectrum of $\bar{F}(\beta)$ in Fig. \ref{fig:incorrectClassification}, where two near-field targets are located at the ranges of  $30$ m and $40$ m  with the Rayleigh distance given by $Z_R = 108 $ m.
It is observed that there exit high peaks associated with near-field targets in the spectrum of $\bar{F}(\beta)$, since the the correlation $|\mathbf{b}(\infty, \theta_3, \phi_3)^H\mathbf{b}(r_{3}, \theta_3, \phi_3)| = 0.96 $ and $|\mathbf{b}(\infty, \theta_4, \phi_4)^H\mathbf{b}(r_{4}, \theta_4, \phi_4)|= 0.98 $, which are strongly correlated.

\underline{\textbf{Proposed$\,$ classification method$\,$ under $\,$half-wavelength}}
\underline{\textbf{inter-antenna spacing:}} 
To address the above issue in existing works on mixed-field classification, we propose to leverage the characteristics of 3D MUSIC spectrum in the range domain to jointly distinguish far-field and near-field targets, estimate ranges of targets, as well as eliminate the ambiguity in angle estimation obtained in Section \eqref{sec:estimation for far-field users}.

To this end, for each candidate target 2D angles, we consider the original covariance matrix $\mathbf{R}_{\boldsymbol{\eta}}$ and obtain its MUSIC spectrum in the range domain given the 2D angles, which is given by
\begin{equation}
	\label{eq:range MUSIC}
	F_g(r, {\theta}_g, {\phi}_g) = \frac{1}{\mathbf{b}^{H}(r, {\theta}_g, {\phi}_g ){\mathbf{U}}_n {\mathbf{U}}_n^{{H}} \mathbf{b}^{H}(r, {\theta}_g,{\phi}_g)},
\end{equation}
where $g\in\mathcal{G}$ denotes the candidate target angle index.
Then, in Fig.~\ref{fig:ProposedClasification}, we numerically show the 3D MUSIC spectrum in the range domain given the 2D angles of the near-field target, far-field target, as well as ambiguous 2D angles. Several important observations are made as follows.
\begin{itemize}
	\item[1.]{\bf{(Near-field target angle)}} When $\{\theta_g, \theta_{g}\}$ are the 2D angles of a near-field target, \emph{its 3D MUSIC spectrum in the range domain will first increases and then decreases, when the range is smaller than the Rayleigh distance} (shown in Fig. \ref{fig:ProposedClasification}), which is referred to as {\bf{Pattern 1}}. In addition, it achieves the peak value in the spectrum $F_g(r, {\theta}_g, {\phi}_g)$ at the target range, which is smaller than the Rayleigh distance $Z_{\rm R}$. 
    (see Fig. \ref{fig:ProposedClasification}). 
	\item[2.]{\bf{(Far-field target angle)}} When $\{\theta_g, \theta_{g}\}$ are the 2D angles of a near-field target, \emph{its 3D MUSIC spectrum in the range domain will monotonically increase when $r\le Z_{\rm R}$}, which is called as {\bf{Pattern 2}}. This is expected since the peak value of the spectrum $F_g(r, {\theta}_g, {\phi}_g)$ is achieved at $ r \to \infty$.
	\item[3.] {\bf{(Ambiguous angle)}} When $\{\theta_g, \theta_{g}\}$ are ambiguous angles, \emph{its 3D MUSIC spectrum in the range domain remains almost static and has no significant peaks/values when $r\le Z_{\rm R}$}, which is terms as {\bf{Pattern 3}}.
    This can be understood based on the MUSIC principle, since only in the true target angles, there will exist MUSIC spectrum  peak due to the orthogonality of signal and noise subspaces.    
\end{itemize}

The above observations indicate that the three types of estimated 2D angles exhibit significantly different patterns in the range-domain MUSIC spectrum. This thus leads to the proposed joint target classification-localization method, as presented below.
\\
\noindent{\footnotesize
\begin{tabular}{|c|c|}
    \hline
    \centering
    $F_g(r,\theta_g,\phi_g) \in $ {\bf{Pattern 1}} & \thead{$(\theta_g,\phi_g)$ are 2D angles of a near-field target.\\ $\hat{r}=\arg\max\limits_{r} F_g(r,\theta_g,\phi_g)$.} \\
    \hline
    $F_g(r,\theta_g,\phi_g) \in $ {\bf{Pattern 2}} &$(\theta_g,\phi_g)$ are 2D angles of a far-field target. \\ \hline
    $F_g(r,\theta_g,\phi_g) \in $ {\bf{Pattern 3}} & $(\theta_g,\phi_g)$ are  ambiguous angles. \\
    \hline
\end{tabular}}


\begin{remark}[Computational complexity]
	\emph{The computational complexity of  proposed mixed-field localization method is dominated by three parts: 1) pre-processing; 2) candidate angle parameter estimation;  3) joint target classification, angle ambiguity elimination and near-field target range estimation.
	In particular, the complexity of pre-processing is dominant in \eqref{eq:inverse}, which amounts to $ \mathcal{O}(LN^3) $, due to the matrix inverse of $ \mathbf{W} $ for $ L $ times.
	For the angle parameter estimation, the dominant complexity mainly involves in the eigenvalue decomposition in \eqref{eq:near-field EVD} that is $ \mathcal{O}\Big(\frac{(N_x+1)^3(N_y+1)^3}{64}\Big) $ and the grid searches for parameter $ \beta_{k} $ and $ \alpha_{k} $ which are $ \mathcal{O}\Big(\ell_y\frac{(N_x+1)^2}{16}+\ell_y\frac{(N_x+1)^3}{8}\Big) $ and  $ \mathcal{O}\Big(2K\ell_xN^2\Big) $, where $\ell_x$ and $\ell_y$ denote the number of search grids for parameters $\alpha$ and $\beta$, respectively~\cite{huang2023low}.
	Then, for the joint target classification, near-field target range estimation, and ambiguity elimination, the dominant complexity is contributed by the grid-search in \eqref{eq:range MUSIC}, given by $ \mathcal{O}(4K\ell_d N^2) $, where $\ell_d$ represents the number of search grids for distance.
	Hence, the overall computational complexity of the proposed localization method is $ \mathcal{O}(K\ell_xN^2 + \ell_yN_x^3 + K\ell_d N^2 + LN^3) $. 	
	As compared to the proposed methods, the 3D MUSIC algorithm incurs a heavy computational complexity of $ \mathcal{O}(\ell_x\ell_y\ell_d N^2 + LN^3) $, since the number of searched grids satisfy $ \ell_x,\ell_y,\ell_d \gg N $.
    Therefore, the proposed decoupled 3D MUSIC-based localization method achieves significant reduction in computational complexity as compared to the 3D MUSIC algorithm.}
\end{remark}
\vspace{-15pt}
\subsection{Extensions and Discussions}
\label{sec:Extensions and Discussions}
\underline{\textbf{Low-complexity MUSIC algorithm:}} When the received signal power of far-field targets is comparable to that near-field targets, we can employ the 2D DFT-based method \cite{fan2018angle} to coarsely estimate the angle parameters for both far-field and near field targets.
As such, we only need to perform grid search in \eqref{eq:far-field beta} and \eqref{eq:far-field alpha estimation} around coarse angles estimated by the DFT-method, which significantly reduces the computational complexity in grid search.

Specifically, we first rearrange the elements in \eqref{eq:anti} to form the following $N_x \times N_y$ matrix
{\small
\begin{equation}
	\mathbf{M} \!\!=\!\! \left[\begin{array}{cccc}
	\!\!\![\tilde{\mathbf{y}}]_1 \!& [\tilde{\mathbf{y}}]_{N_x+1} \!& \ldots \!\!& \!\!\![\tilde{\mathbf{y}}]_{(N_y-1)N_x+1}\!\! \\
	\!\!\![\tilde{\mathbf{y}}]_2 \!& [\tilde{\mathbf{y}}]_{N_x+2} \!& \ldots \!\!& \!\!\![\tilde{\mathbf{y}}]_{(N_y-1)N_x+2}\!\! \\
	\!\!\!\vdots \!& \vdots \!& \ddots \!& \vdots\!\! \\
	\!\!\![\tilde{\mathbf{y}}]_{N_x} \!& [\tilde{\mathbf{y}}]_{2N_x} \!& \ldots \!\!& \!\!\![\tilde{\mathbf{y}}]_{N_yN_x}\!\!
	\end{array}\right] .
\end{equation}}
In particular, the $(i,j)$-th entry of the matrix $\mathbf{M}$ is given by 
\begin{equation}
	[\mathbf{M}]_{i,j} \!\!=\!\! \sum_{k=1}^K g_k e^{\jmath 2d_0(i\alpha_k \!+\! j\beta_k)}e^{-\jmath 2d_0(\frac{N_x+1}{2}\alpha_k \!+\! \frac{N_y+1}{2}\beta_k)},
\end{equation} 
where $i = 1,2,\cdots,N_x$ and $j = 1,2,\cdots,N_y$.
Then, we define two DFT matrices $\mathbf{D}_x \in \mathbb{C}^{N_x \times N_x}$ and $\mathbf{D}_y \in \mathbb{C}^{N_y \times N_y}$, which are given by
$[\mathbf{D}_x]_{i_x,j_x} = \frac{1}{N_x}e^{-\jmath2\pi \frac{(i_x-1)(j_x-1)}{N_x}}, i_x,j_x = 1,2,\cdots,N_x $ and
$[\mathbf{D}_y]_{i_y,j_y} = \frac{1}{N_y}e^{-\jmath2\pi \frac{(i_y-1)(j_y-1)}{N_y}}, i_y,j_y = 1,2,\cdots,N_y. $
As such, the normalized 2D-DFT of the constructed matrix $\mathbf{M}$ can be calculated as
\begin{align}
	\big|&[\mathbf{M}_{\rm DFT}]_{{i,j}}\big| = \big|[\mathbf{D}_x\mathbf{M}\mathbf{D}_y]_{{i,j}}\big|\notag\\
	&=\frac{1}{N_x N_y} \Bigg|\sum_{i_x = 1}^{N_x} \sum_{i_y=1}^{N_y}[\mathbf{M}]_{{i_x, i_y}}\mathrm{e}^{-\jmath 2 \pi\left(\frac{(i_x-1)i}{N_x}+\frac{(i_y-1)j}{N_y}\right)}\Bigg|\notag\\
	&\triangleq\!\! \frac{1}{N_x N_y}\Bigg|\sum_{k=1}^{K}g_k\Xi_{N_x}(\alpha_k - {i}/{N_x})\Xi_{N_y}(\beta_k - {j}/{N_y})\Bigg|,
\end{align}
where $\Xi_{N}(x) \triangleq \frac{\sin({\pi N x})}{\sin({\pi x})}$ is the Dirichlet Sinc function. 

It can be observed that when $(\alpha_k - {i}/{N_x})\pi = \ell_x\pi, \forall \ell_x \in \mathbb{Z}$ and $(\beta_k - {j}/{N_y})\pi = \ell_y\pi, \forall \ell_y \in \mathbb{Z} $ (i.e., $i = N_x(\alpha_k-1)$ and $j = N_y(\beta_k-1)$),  $\big|\mathbf{M}_{\rm DFT}[{i,j}]\big|$ forms a peak with a height $|g_k|$.
Hence, by denoting ($\hat{i}_k$, $\hat{j}_k$) as an integer pair where $\big|[\mathbf{M}_{\rm DFT}]_{{\hat{i}_k,\hat{j}_k}}\big|$ reaches a peak value, the parameter estimation for the $k$-th target is given by $\hat{\alpha}_k^{(1)} = \frac{\hat{i}_k}{N_x}, \hat{\alpha}_k^{(2)} = \frac{\hat{i}_k}{N_x}-1$, and $\hat{\beta}_k^{(1)} = \frac{\hat{j}_k}{N_y}, \hat{\beta}_k^{(2)} = \frac{\hat{j}_k}{N_y}-1$,
where $\hat{\alpha}_k^{(1)}$ and $\hat{\alpha}_k^{(2)}$ are estimated parameter angles for the $k$-th target, corresponding to the true and ambiguous angle parameters, respectively.
Similarly, $\hat{\beta}_k^{(1)}$ and $\hat{\beta}_k^{(2)}$ follow the same pattern as $\hat{\alpha}_k^{(1)}$ and $\hat{\alpha}_k^{(2)}$. As such, we  obtain the coarse angle parameters $ \{\hat{\alpha}_k^{(1)},\hat{\alpha}_k^{(2)},\hat{\beta}_k^{(1)},\hat{\beta}_k^{(2)}, \forall k \in \mathcal{K}\} $ with a resolution $1/N_x$ and $1/N_y$ for the parameter $\alpha$ and $\beta$ estimation, respectively.
To facilitate understanding, in Fig. \ref{fig:DFT}, we numerically plot an example of 2D-DFT of the constructed matrix $\mathbf{M}$ with $N_x = N_y = 101 $.

Then, for the grid searches in the estimators \eqref{eq:far-field beta} and \eqref{eq:far-field alpha estimation}, we only need to {refine angle estimation in the range $\{ [\hat{\alpha}_k^{(1)}-\frac{1}{N_x}, \hat{\alpha}_k^{(1)}+\frac{1}{N_x}],[\hat{\beta}_k^{(1)}-\frac{1}{N_y}, \hat{\beta}_k^{(1)}+\frac{1}{N_y}],\forall k \in \mathcal{K}\} $,} resulting in {computational complexity decreasing to $ \mathcal{O}(K\ell_xNN_y + \ell_yN_x^3/N_y + K\ell_d N^2 + LN^3) $.}

However, since the height of the peaks in $\mathbf{M}_{\rm DFT}$ depends on the signal powers $|g_k|$, when the received signal powers at the BS from far-field targets are relatively weak, their peaks in $\mathbf{M}_{\rm DFT}$ may be obscured by the side-lobes of near-field targets with stronger DFT-lobe signal power. In this case, the coarse estimation method based on 2D-DFT cannot be used to reduce the computational complexity of the proposed algorithm. How to overcome this issue will be left for our future work.
\begin{figure}[t]
	\vspace{-14pt}
	\includegraphics[width=0.4\textwidth]{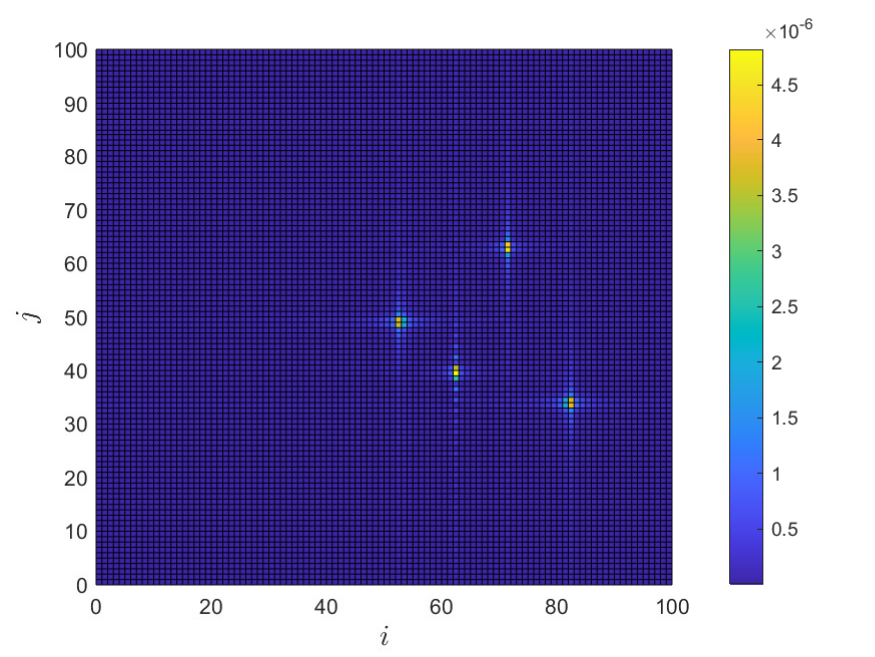}
	\centering
	\caption{An example of normalized 2D-DFT of the constructed matrix $\mathbf{M}$. System parameters are the same as that in Fig. \ref{fig:ProposedClasification}. Taking one peak index ($\hat{i} = 62$, $\hat{j} = 39$) as an example (4 peaks in total), the estimated parameter $\alpha$ are $\hat{\alpha}^{(1)} = \frac{62}{101} = 0.6139$ and $\hat{\alpha}^{(2)} = \frac{62}{101} - 1 = -0.3861$, while parameter $\beta$ are estimated by $\hat{\beta}^{(1)} = \frac{39}{101} = 0.3861$ and $\hat{\beta}^{(2)} = \frac{39}{101} -1 = -0.6139$.
	The estimated angle parameters are corresponding to the near-field target located at ($-\frac{5\pi}{13}$, ${0.23\pi}$, $50$ m) where the parameter $\alpha = \sin(-\frac{5\pi}{13})\cos(0.23\pi) = -0.7014 $ and $\beta = \sin(-\frac{5\pi}{13})\sin(0.23\pi) = -0.6183 $.
	Obviously, the estimated parameter $\hat{\alpha}^{(1)}$ and $\hat{\beta}^{(1)}$ are corresponding ambiguous angles.}
	\label{fig:DFT}
	\vspace{-16pt}
\end{figure}

\underline{\textbf{Active sensing scenarios:}}
The proposed mixed near-field and far-field target localization method can also be extended to the active-sensing scenario where the BS estimates the target locations based on echo signals transmitted by itself \cite{Wang2023NFISAC}.
For this case, let $\mathbf{x}_{\rm t}(m) = \mathbf{W}_{\rm t}\mathbf{s}(m) \in \mathbb{C}^{N \times 1}$ denote the transmitted signal by the BS at time $m$, where $\mathbf{W}_{\rm t} \in \mathbb{C}^{N \times N_{\rm RF}}$ represents the transmit combining matrix.
Then, the echo signal (reflected by targets) received by the BS at discrete time $m$ is given by $\mathbf{y}(m) = \mathbf{W}_{\rm r}(m)\mathbf{G}\mathbf{x}(m) +  \mathbf{W}_{\rm r}(m)\mathbf{n}(m)$,
where $\mathbf{W}_{\rm r}(m)$ represents receive combining matrix at time $m$ and $\mathbf{x}(m) = [x_1(m),x_2(m),\cdots,x_K(m)]^T \in \mathbb{C}^{K \times 1} $ denotes the equivalent transmitted signal vector of the $K$ targets.
In particular, $x_k(m) \in \mathbb{C}$ is the equivalent transmitted signal at $k$-th target, given by
\begin{equation}
	x_k(m) = \gamma_k\mathbf{b}^H(r_k, \theta_k, \phi_k)\mathbf{W}_{\rm t}\mathbf{s}(m),~~\forall k\in\mathcal{K},
\end{equation}
where $r_k \to \infty, \forall k \in \mathcal{K}_1$ and $\gamma_k$ denotes the reflection coefficient, which contains both
the round-trip path-loss and the radar cross section (RCS).

As such, similar to Section \ref{sec:signal reconstruction}, the received signal at the BS UPA can be reconstructed as $\hat{\mathbf{y}} = \mathbf{G} \mathbf{x} + \mathbf{n}$.
Hence, the covariance matrix of signal $\hat{\mathbf{y}}$ is given by
\begin{equation}
	\label{eq:echo covariance matrix}
	\mathbf{R} = \mathbb{E}\!\left\{\hat{\mathbf{y}} \hat{\mathbf{y}}^{{H}}\right\} = \mathbf{G}\mathbf{R}_x\mathbf{G}^H + \sigma^2 \mathbf{I}_{N},
\end{equation}
where $\mathbf{R}_x = \mathbb{E}\{{\mathbf{x}} {\mathbf{x}}^{H}\} $ denotes the source covariance matrix.
Mathematically, the ($i$, $j$) entry of $\mathbf{R}_x$ can be obtained as
\begin{equation}
	{[\mathbf{R}_x]_{i,j}} = \gamma_i\gamma_j^\ast\mathbf{b}^H(r_i, \theta_i, \phi_i)\mathbf{W}_{\rm t}\mathbf{W}_{\rm t}^H\mathbf{b}(r_j, \theta_j, \phi_j),
\end{equation}
where $r_i,r_j \to \infty $, if $i,j \in \mathcal{K}_2$.
By carefully designing the transmit combining matrix $\mathbf{W}_{\rm t}$, we can achieve $\mathbf{W}_{\rm t}\mathbf{W}_{\rm t}^H = \diag\{\mathbf{I}_{\rm RF}, \mathbf{0}_{N-N_{\rm RF}}\}$ such that $\mathbf{R}_x[i,j]$ can be rewritten as $\mathbf{R}_x[i,j] = \widetilde{\mathbf{b}}^H(r_i, \theta_i, \phi_i)\widetilde{\mathbf{b}}(r_i, \theta_i, \phi_i) $, where $\widetilde{\mathbf{b}}^H(r_i, \theta_i, \phi_i) = [{\mathbf{b}}^H(r_i, \theta_i, \phi_i)]_{1:N_{\rm RF}}$ and $\widetilde{\mathbf{b}}^H(r_j, \theta_j, \phi_j) = [{\mathbf{b}}^H(r_j, \theta_j, \phi_j)]_{1:N_{\rm RF}} $.

It is worth noting that when targets are widely separated and $N_{\rm RF}$ is large enough, the correlation between $\widetilde{\mathbf{b}}^H(r_i, \theta_i, \phi_i)$ and $\widetilde{\mathbf{b}}^H(r_j, \theta_j, \phi_j)$ ($i \neq j$) is negligible. As a result, $\mathbf{R}_x$ is approximately a \emph{diagonal} matrix, i.e., $\mathbf{R}_x \approx \diag\{|\gamma_1|^2N_{\rm RF},\cdots,|\gamma_k|^2N_{\rm RF},\cdots,|\gamma_K|^2N_{\rm RF}\}$.
Hence, the original covariance matrix of the echo signals in \eqref{eq:echo covariance matrix} has the same form as the covariance matrix of signals in \eqref{eq:orignal covariance matrix}, for which our proposed mixed-field target localization method can be directly applied.
However, when channel correlation between different targets is large enough, the anti-diagonal elements of $\mathbf{R}$ will be a linear combination of the signals received by multiple antennas, leading to erroneous peaks in the corresponding angular spectrum.
How to address this issue will be left for our future work.

\section{Cram\' er-Rao bound derivation}
In this section, we characterize the CRB for the mixed-field parameter estimation for evaluating the performance of the proposed localization algorithm.

Specifically, the stochastic CRB provides a lower bound on the covariance of an arbitrary unbiased estimator $\hat{\boldsymbol{\eta}}$, where $\hat{\boldsymbol{\eta}}$ denotes the estimated parameter vector.
Mathematically, 
\begin{equation}
	\mathrm{MSE}([\hat{\boldsymbol{\eta}}]_{i}) \!=\! {\mathbb{E}}\Big\{\left([\hat{\boldsymbol{\eta}}]_{i}\!-\![\boldsymbol{\eta}]_{i}\right)^{2}\Big\} \!\geq\! \left[\mathbf{CRB}(\boldsymbol{\eta})\right]_{i,i}, \forall i \in \mathcal{I},
	\notag
\end{equation}
where $ \boldsymbol{\eta} $ represents the parameter vector and $ \mathbf{CRB}(\boldsymbol{\eta}) $ denotes the corresponding CRB matrix, which can be obtained by the inversion of the Fisher information matrix (FIM).
In addition, $\mathcal{I} \triangleq \{1,2,\cdots, 2K+K_2\}$ denotes the index set for parameters to be estimated.
In this paper, the parameter vector of interest is given by
\begin{equation}
	\begin{aligned}
		\boldsymbol{\eta} = &[\theta_1, \cdots, \theta_{K_1},\phi_1, \cdots, \phi_{K_1}, \theta_{K_1+1}, \cdots, \theta_{K},\\
		&~~~~~~~~~~~~~~~~~~\phi_{K_1+1}, \cdots, \phi_{K}, r_{K_1+1}, \cdots, r_{K}]^T.
	\end{aligned}
\end{equation}
Moreover, the parameters, such as noise power $ \sigma^2 $ and the entries in $ \mathbf{S} = \mathbb{E}\{\mathbf{s}(t)\mathbf{s}^H(t)\} $ (including imaginary and real parts), are regarded as nuisance parameters.
The signal model is based on the reconstructed signals received at UPA in \eqref{eq:inverse}. 

Based on the above, the $ (i,j) $-th entry of the stochastic CRB matrix $ \mathbf{CRB}(\boldsymbol{\eta}) \in \mathbb{C}^{(2K+K_2) \times (2K+K_2)} $ is given by 
\begin{equation}
	\left[\mathbf{CRB}^{-1}(\boldsymbol{\eta})\right]_{i, j} \!\!=\!\! \frac{2 L}{\sigma^2} \operatorname{Re}\!\left\{\operatorname{Tr}\!\left[\frac{\partial \mathbf{G}^H}{\partial [\boldsymbol{\eta}]_{i}} \Pi_{\boldsymbol{G}}^{\perp} \frac{\partial \mathbf{G}}{\partial [\boldsymbol{\eta}]_j} \mathbf{Q}^T\right]\right\},
\end{equation}
where $ \Pi_{\mathbf{G}}^{\perp} $ denotes the orthogonal projection matrix onto the null space of the ASM $ \mathbf{G} $.
Moreover, $ \Pi_{\mathbf{G}}^{\perp}\in \mathrm{C}^{N \times N} $ and $ \mathbf{Q} \in \mathrm{C}^{K \times K} $ are respectively given by
\begin{equation}
	\Pi_{\mathbf{G}}^{\perp} = \mathbf{I}_{N} - \mathbf{G}(\mathbf{G}^H\mathbf{G})^{-1}\mathbf{G}^H,
\end{equation}
\begin{equation}
	\mathbf{Q}=\mathbf{S} \mathbf{G}^H \mathbf{R}^{-1} \mathbf{G}\mathbf{S}.
\end{equation}
To obtain the closed-form CRB matrix, we first form the FIM $ \mathbf{CRB}^{-1}(\boldsymbol{\eta}) $ as follows 
\begin{equation}
		\mathbf{CRB}^{-1}(\boldsymbol{\eta})= \mathbf{F} \triangleq \frac{2L}{\sigma^2}
		\left[\begin{array}{lll}
		\mathbf{F}_{\rm FF} & \mathbf{F}_{\rm FN} \\
		\mathbf{F}_{\rm NF} & \mathbf{F}_{\rm NN}
	\end{array}\right],
\end{equation} 
where $ \mathbf{F}_{\rm FF} \in \mathrm{C}^{2K_1 \times 2K_1} $ and $ \mathbf{F}_{\rm NN} \in \mathrm{C}^{3K_2 \times 3K_2} $ denote the FIM partitions of far-field and near-field parameters.
In addition, $ \mathbf{F}_{\rm FN} \in \mathrm{C}^{2K_1 \times 3K_2 } $ and $ \mathbf{F}_{\rm NF} \in \mathrm{C}^{3K_2 \times 2K_1 } $ represent the mutual FIM partitions. Then, by partitioning $ \mathbf{Q} $ from the sub-matrices $ \mathbf{Q}_1 \in \mathrm{C}^{K_1 \times K_1} $, $ \mathbf{Q}_2 \in \mathrm{C}^{K_1 \times K_2} $, $ \mathbf{Q}_3 \in \mathrm{C}^{K_2 \times K_1} $ and $ \mathbf{Q}_4 \in \mathrm{C}^{K_2 \times K_2} $, $ \mathbf{Q} $ can be rewritten as
\begin{equation}
	\mathbf{Q} = 
	\left[\begin{array}{lll}
		\mathbf{Q}_{1} & \mathbf{Q}_{2} \\
		\mathbf{Q}_{3} & \mathbf{Q}_{4}
	\end{array}\right].
\end{equation}
Similar to the manipulations in \cite{stoica2001stochastic}, the sub-matrices of the FIM can be obtained as 
\begin{equation}
	\mathbf{F}_{\rm FF}=\left\{\operatorname{Re}\left[\left(\mathbf{D}_{\mathrm{F}}^H \Pi_{\mathbf{G}}^{\perp} \mathbf{D}_{\mathrm{F}}\right) \odot (\mathbf{J}_{\rm FF} \otimes\mathbf{Q}_1^T)\right]\right\},
\end{equation}
\begin{equation}
	\mathbf{F}_{\rm FN}=\left\{\operatorname{Re}\left[\left(\mathbf{D}_{\mathrm{F}}^H \Pi_{\mathbf{G}}^{\perp} \mathbf{D}_{\mathrm{N}}\right) \odot (\mathbf{J}_{\rm FN} \otimes\mathbf{Q}_3^T)\right]\right\},
\end{equation}
\begin{equation}
	\mathbf{F}_{\rm NF}=\left\{\operatorname{Re}\left[\left(\mathbf{D}_{\mathrm{N}}^H \Pi_{\mathbf{G}}^{\perp} \mathbf{D}_{\mathrm{F}}\right) \odot (\mathbf{J}_{\rm NF} \otimes\mathbf{Q}_2^T)\right]\right\},
\end{equation}
\begin{equation}
	\mathbf{F}_{\rm NN}=\left\{\operatorname{Re}\left[\left(\mathbf{D}_{\mathrm{N}}^H \Pi_{\mathbf{G}}^{\perp} \mathbf{D}_{\mathrm{N}}\right) \odot (\mathbf{J}_{\rm NN} \otimes\mathbf{Q}_4^T)\right]\right\},
\end{equation}
where we denote 
\begin{equation}
	\mathbf{D}_{\mathrm{F}} = \left[\sum_{k_1=1}^{K_1} \frac{\partial \mathbf{A}}{\partial \theta_{k_1}}, \sum_{k_1=1}^{K_1}\frac{\partial \mathbf{A}}{\partial \phi_{k_1}}\right],
\end{equation}
\begin{equation}
	\mathbf{D}_{\mathrm{N}} \!=\! \left[\sum_{k_2=K_1+1}^{K} \frac{\partial \mathbf{B}}{\partial \theta_{k_2}}, \!\sum_{k_2=K_1+1}^{K}\frac{\partial \mathbf{B}}{\partial \phi_{k_2}}, \!\sum_{k_2=K_1+1}^{K}\frac{\partial \mathbf{B}}{\partial r_{k_2}}\right],
\end{equation}
\begin{equation}
	\mathbf{J}_{\rm FF} = [1,1]\otimes[1,1]^T,
\end{equation}
\begin{equation}
	\mathbf{J}_{\rm FN} = [1,1,1]\otimes[1,1]^T,
\end{equation}
\begin{equation}
	\mathbf{J}_{\rm NF} = [1,1]\otimes[1,1,1]^T,
\end{equation}
\begin{equation}
	\mathbf{J}_{\rm NN} = [1,1,1]\otimes[1,1,1]^T.
\end{equation}

As such, we  obtain the \emph{matrix closed-form} CRB for the mixed-field localization.
In particular, when all the targets are located in the far or near-field, the CRB matrices  reduces to $ \mathbf{CRB}(\boldsymbol{\eta}) = \mathbf{F}_{\rm FF}^{-1} $ and $ \mathbf{CRB}(\boldsymbol{\eta}) = \mathbf{F}_{\rm NN}^{-1} $, respectively.

\vspace{-5pt}
\section{Numerical Results}
In this section, numerical results are presented to demonstrate the effectiveness of our proposed mixed near-field and far-field target localization method.

\vspace{-10pt}
\subsection{System Setup and Benchmark Schemes}
System parameters are set as follows, unless otherwise specified.
The BS operates at the $f=10$ GHz frequency band, and the UPA at the BS is equipped with $61 \times 61$ antennas, for which the Rayleigh distance can be calculated as $ Z_{\rm R} =2\times(61^2\times0.015^2+61^2\times0.015^2)/{0.03} \approx 112$ m.
In addition, the number of RF chains is $N_{\rm} = 61$ with $M_x = 61$ and $M_y=1$.
The number of equivalent snapshots is set as $L = 500$.
There are four targets in the setup: 1) far-field target 1 located at $(\frac{\pi}{4},-\frac{\pi}{3})$ rad; 2) far-field target 2 located at $(\frac{\pi}{8},\frac{\pi}{3})$ rad; near-field target 3 located at ($\frac{\pi}{4}$ rad,$\frac{\pi}{4}$ rad, $30$ m); near-field target 4 located at ($\frac{\pi}{8}$ rad,$\frac{\pi}{4}$ rad, $40$ m).
Moreover, we consider the RMSE for angle and range estimation for performance evaluation, which are defined as {\small$\theta_{\mathrm{RMSE}} = \Big({\frac{1}{Q K} \sum_{q=1}^Q \sum_{k=1}^K\big(\hat{\theta}_k^{(q)}-\theta_k\big)^2}\Big)^{\frac{1}{2}}$}, {\small$\phi_{\mathrm{RMSE}} = \Big({\frac{1}{Q K} \sum_{q=1}^Q \sum_{k=1}^K\big(\hat{\phi}_k^{(q)}-\phi_k\big)^2}\Big)^{\frac{1}{2}}$} and {\small$r_{\mathrm{RMSE}} = \Big({\frac{1}{Q K} \sum_{q=1}^Q \sum_{k=1}^K\big(\hat{r}_k^{(q)}-r_k\big)^2}\Big)^{\frac{1}{2}}$}, where $\hat{\theta}_k^{(q)}$, $\hat{\phi}_k^{(q)}$ and $\hat{r}_k^{(q)}$ denote the estimated 
elevation and azimuth angles and range for the $k$-th target in the $q$-th Monte Carlo simulation over $Q = 100$ Monte Carlo simulations.
For performance comparison, we consider the following benchmark schemes:
\begin{itemize}
	\item {\rm \textbf{3D MUSIC algorithm:}}
	This scheme performs joint searches over the elevation angle, azimuth angle and range domains in the spatial spectrum.
	\item {\rm \textbf{Proposed method + DFT:}}
	For this scheme, we perform coarse angle estimation by the 2D-DFT method in Section \ref{sec:Extensions and Discussions} and then refine the angle estimation via 1D grid searches in small ranges, which significantly reduces the computational complexity.
	\item {\rm \textbf{Localization method ${d_0 = \frac{\lambda}{4}} $}:} This scheme is similar to our proposed method. It does has the angular ambiguity issue, but it has smaller array aperture and hence potentially reduced localization accuracy. 
    \item {\rm \textbf{Root CRB:}} Root CRB (RCRB) is used to serve as the estimation performance lower bound. 
\end{itemize}
\begin{figure}[t]
	\vspace{-14pt}
	\includegraphics[width=0.4\textwidth]{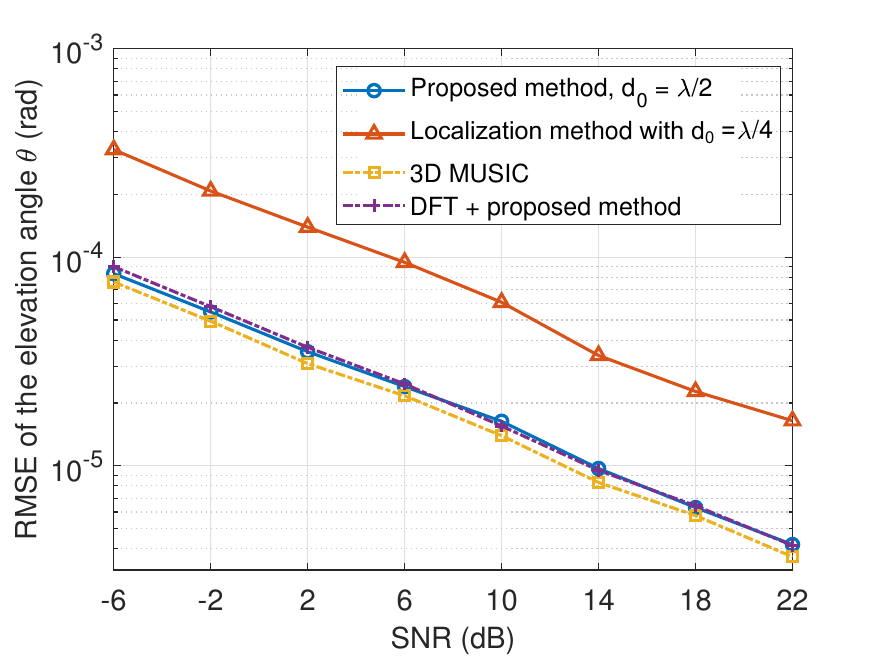}
	\centering
	\caption{Elevation angle RMSE versus SNR.}
	\label{fig:RMSE_theta}
	\vspace{-10pt}
\end{figure}
\begin{figure}[t]
	\vspace{-8pt}
	\includegraphics[width=0.4\textwidth]{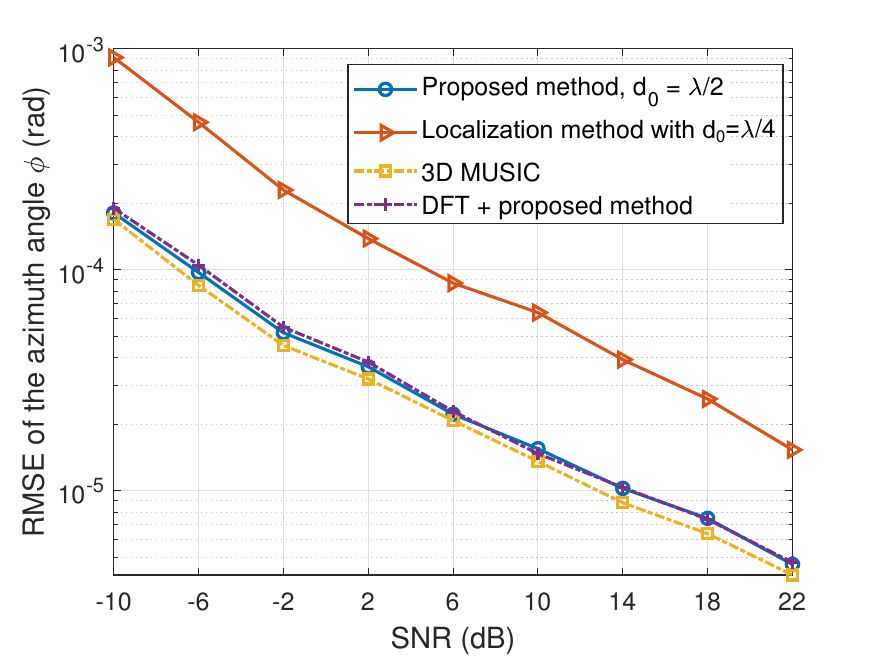}
	\centering
	\caption{Azimuth angle RMSE versus SNR.}
	\label{fig:RMSE_phi}
	\vspace{-16pt}
\end{figure}
\begin{figure}[t]
\vspace{-14pt}
	\includegraphics[width=0.4\textwidth]{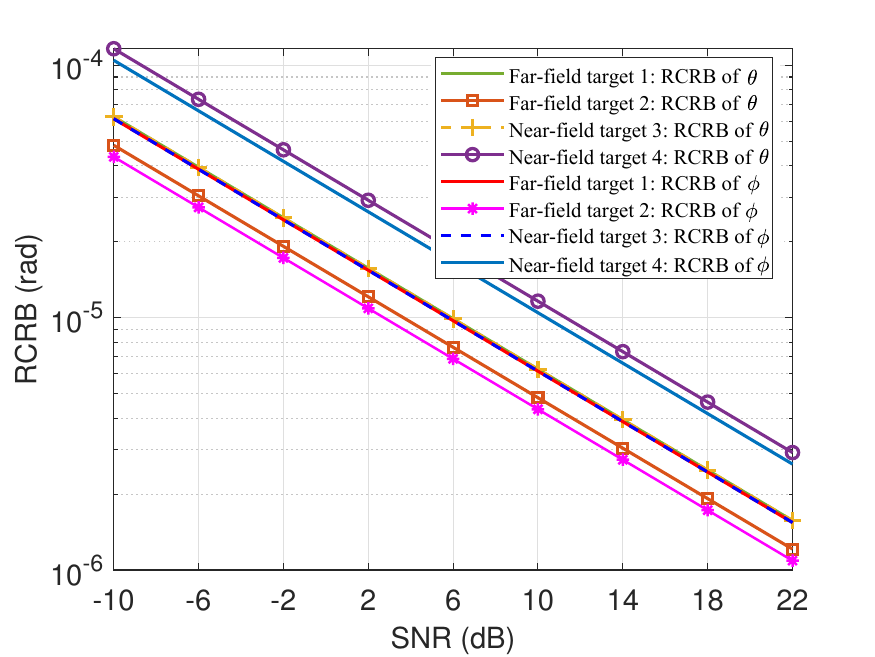}
	\centering
	\caption{Angle RCRB versus SNR.}
	\label{fig:CRB}
	\vspace{-10pt}
\end{figure}
\begin{figure}[t]
	\vspace{-6pt}
	\includegraphics[width=0.4\textwidth]{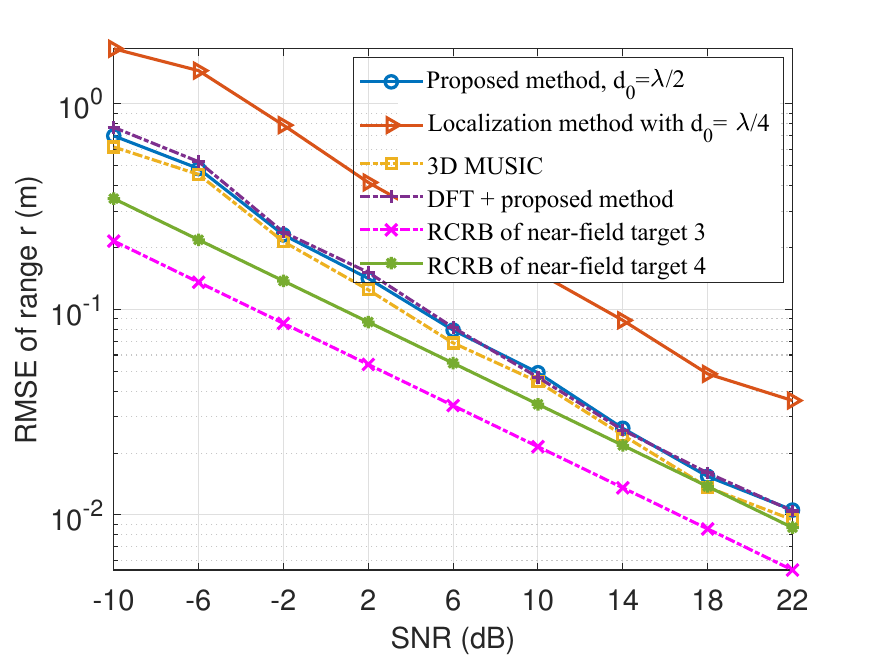}
	\centering
	\caption{Range RMSE versus SNR.}
	\label{fig:RMSE_r}
	\vspace{-18pt}
\end{figure}
\begin{table*}[t]
	\centering
	{
		\caption{Target classification under different methods.}
		\label{Table:classification method}
		\renewcommand{\arraystretch}{1.75}
		\begin{tabular}{|c|c|c|c|c|c|} 
			\hline
			\multicolumn{2}{|c|}{ \diagbox{Method}{Target index}}& \thead{Far-field target 1\\$(\frac{\pi}{4},\frac{\pi}{3},-)$} &  \thead{Far-field target 2\\ $(\frac{\pi}{8},\frac{\pi}{3},-)$} & \thead{Near-field target 3\\($\frac{\pi}{4}$,$\frac{\pi}{4}$, $30$ m)} & \thead{Near-field target 4\\($\frac{\pi}{8}$,$\frac{\pi}{4}$, $40$ m)}          
            \\ \hline
			\multirow{3}{*}{\text{Proposed method}} & $\theta$ & 0.78540 & 0.39270 &0.78539 &0.39270\\ \cline{2-6}
            & $\phi$ & 1.04720 & 1.04719 &0.78540 &0.78540\\ \cline{2-6}
			& $r$ & - & - &30.03 &40.06\\ \cline{1-6}
			\multirow{3}{*}{\cite{he2011efficient}} & $\theta$ & 0.78540 & 0.39268 &0.78540 &0.39269\\ \cline{2-6}
            & $\phi$ & 1.04720 & 1.04719 &0.78538 &0.78540\\ \cline{2-6}
			& $r$ & - & - &\XSolidBrush &\XSolidBrush\\ \cline{1-6}	
	\end{tabular}}
\end{table*}
\vspace{-10pt}
\subsection{Performance Analysis}
In Table \ref{Table:classification method}, we summarize the effectiveness of different mixed-field classification methods, where the reference SNR of each target (defined by $\frac{N|\beta_k|^2p_k}{r_k^2\sigma^2}$) is set as 10 dB.
It is observed that the proposed method and the benchmark scheme in \cite{he2011efficient} achieve similar accuracy in angle estimation.
However, the method proposed in \cite{he2011efficient} wrongly estimate the near-field targets as far-field targets and thus cannot estimate the ranges of near-field targets (see Section \ref{sec:Extensions and Discussions}).

In Fig. \ref{fig:RMSE_theta} and Fig. \ref{fig:RMSE_phi}, we plot the elevation and azimuth RMSE versus the reference SNR.
It is observed that the proposed method with $d=\lambda/2$ outperforms the case with $d=\lambda/4$. This is because, with the same number of antennas, the considered half-wavelength antenna spacing case can form a larger aperture, thereby enhancing sensing accuracy.
Moreover, our proposed mixed-field localization method exhibits similar RMSE performance to the counterpart of the 3D MUSIC method, demonstrating the effectiveness of our proposed method. 
In addition, the scheme by employing the 2D-DFT method performs slightly worse than the proposed method at the low-SNR regime, while achieving nearly the same RMSE performance at the high-SNR regime.
This is because the peak values searched by the 2D-DFT method are proportional to the SNR and the 2D-DFT method is prone to misjudgments at the low-SNR regime.
By comparing Figs. \ref{fig:RMSE_theta}--\ref{fig:CRB}, it can be observed that the angle RMSE of our proposed mixed-field localization method achieves the same order of the RCRB.

In Fig. \ref{fig:RMSE_r}, we plot the range RMSE versus SNR.
It is observed that the proposed mixed-field localization method achieves nearly the same performance as the 3D MUSIC algorithm, while the proposed method requires much lower computational complexity.
Moreover, although the scheme with a quarter-wavelength inter-antenna spacing does not suffer from angle ambiguity, our scheme with $d_0 = \frac{\lambda}{2}$ still achieves much high estimation accuracy.
This phenomenon arises because the UPA with $d_0 = \frac{\lambda}{2}$ has a larger array aperture as compared to the counterpart with $d_0 = \frac{\lambda}{4}$, resulting in a higher estimation resolution.
In addition,  Fig. \ref{fig:RMSE_r} shows that the RMSE performance of the proposed mixed-field target localization method is reasonably close to the RCRB.

Finally, we plot the running time versus the number of antennas in Fig. \ref{fig:running time}, where the numbers of search grids are set as $\ell_x = \ell_y = 10^{4}$ and $\ell_d = 10^{3} $.
It is observed that the proposed method as well as the ``proposed method + DFT'' significantly reduce the computational complexity.
This is because the proposed method decouple the angle and range estimation, while the 2D-DFT method first coarsely estimates target angles and then refine angle estimation by applying 1D MUSIC algorithm, which narrows the search regime.
\begin{figure}[t]
	\vspace{-14pt}
	\includegraphics[width=0.4\textwidth]{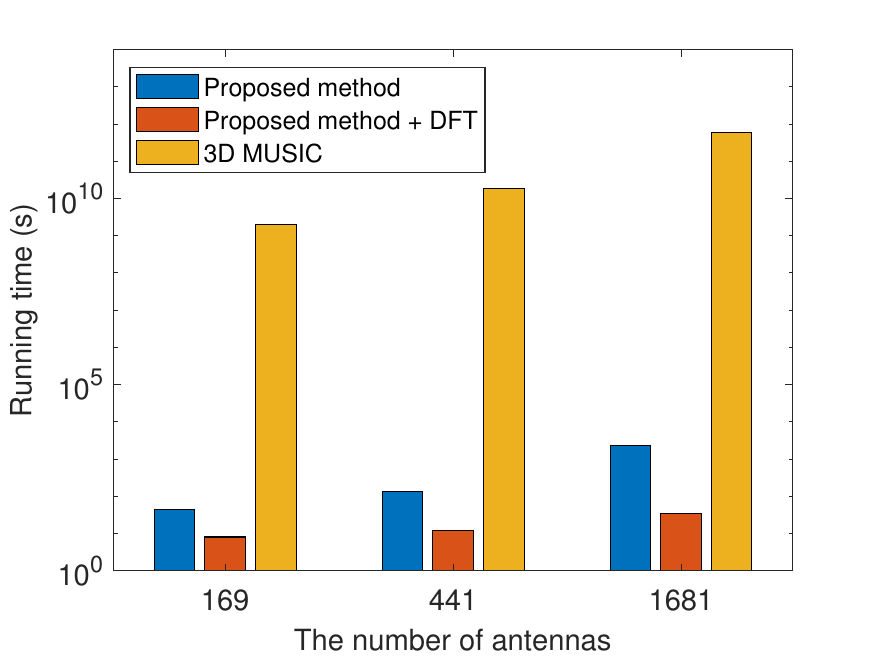}
	\centering
	\caption{Running time versus the number of antennas.}
	\label{fig:running time}
	\vspace{-16pt}
\end{figure}	
\section{Conclusions}
In this paper, we proposed a three-step mixed near-field and far-field target localization method by employing typical wireless communication infrastructures.
First, we judiciously designed the analog combining matrix at the BS over time to recover the received signals at antennas with minimum recovery errors.
Then, we decoupled the angle and range estimation problems via extracting the anti-anti-diagonal elements of the original covariance matrix, for which a decoupled 2D MUSIC algorithm was proposed to estimate target angles including the true far-field and near-field angles as well as their corresponding ambiguous angles.
Further, an effective classification method was proposed to distinguish mixed-field targets and resolve the angular ambiguity based on different patterns in the range-domain MUSIC spectrum of the three types of angles.
Finally, we derived the CRB for mixed-field target localization as the estimation error lower bound and numerical results were presented to demonstrate the effectiveness of the proposed method as compared with various benchmark schemes.
\bibliographystyle{IEEEtran}
\bibliography{IEEEabrv.bib}
\end{document}